\documentclass[draftcls,onecolumn]{IEEEtran}

\usepackage{amsmath, amssymb, amsthm}
\usepackage{graphicx, soul, color}
\usepackage{algorithm, algorithmic}
\usepackage{hyperref}
\usepackage{cite}
\usepackage{dsfont}
\usepackage{enumitem}
\usepackage{bm}
\usepackage[normalem]{ulem}
\usepackage{subcaption}
\usepackage{multirow}
\usepackage{comment}

\setlist{itemsep=4pt}

\renewcommand\vec[1]{\boldsymbol{#1}}
\newcommand\mat[1]{\boldsymbol{#1}}
\newcommand\trans{T}
\newcommand\norm[1]{\left\| #1 \right\|}

\DeclareMathOperator{\diag}{diag}
\DeclareMathOperator{\cone}{cone}
\DeclareMathOperator{\vol}{Vol}

\newtheorem{theorem}{Theorem}
\newtheorem{proposition}{Proposition}
\newtheorem{lemma}{Lemma}
\newtheorem{corollary}{Corollary}

\newtheorem{property}{Property}
\newtheorem{definition}{Definition}

\newif\ifhideproofs

\ifhideproofs
    \usepackage{environ}
    \NewEnviron{hide}{}

\fi

\begin{document}

\title{Identifiable Solutions to Foreground Signature Extraction from Hyperspectral Images in an Intimate Mixing Scenario}

\author{Jarrod Hollis, Raviv Raich, Jinsub Kim, Barak Fishbain, and Shai Kendler \thanks{This paper extends our preliminary work presented in \emph{Hollis J, Raich R, Kim J, Fishbain B, Kendler S, ``Foreground signature extraction for an intimate mixing model in hyperspectral image classification'', IEEE Intl.~Conf.~on Acoustics, Speech and Signal Processing (ICASSP) 2020, pp. 4732-4736.} \cite{hollis2020foreground} This work was partially supported by the Israel Ministry of Science and Technology Research Program, the Israel Ministry of Environmental Protection.}}

\markboth{}%
{Hollis \MakeLowercase{\textit{et al.}}: Identifiable Solutions to Foreground Signature Extraction for a Hyperspectral Intimate Mixing Model}

\maketitle

\begin{abstract}
The problem of foreground material signature extraction in an intimate (nonlinear) mixing setting is considered. It is possible for a foreground material signature to appear in combination with multiple background material signatures. We explore a framework for foreground material signature extraction based on a patch model that accounts for such background variation. We identify data conditions under which a foreground material signature can be extracted up to scaling and elementwise-inverse variations. We present algorithms based on volume minimization and endpoint member identification to recover foreground material signatures under these conditions. Numerical experiments on real and synthetic data illustrate the efficacy of the proposed algorithms. 
\end{abstract}


\IEEEpeerreviewmaketitle

\section{Introduction}

\IEEEPARstart{T}{he} hyperspectral unmixing problem has applications in endmember signature extraction \cite{winter1999n, dobigeon2009joint} and classification \cite{ghamisi2017advanced, camps2005kernel} tasks. A pixel from a hyperspectral image may be viewed as a mixture of the spectral signatures from materials located within the spatial bounds of the pixel. In remote sensing, mixtures are commonly described as non-negative linear combinations of spectral signatures weighted by the proportions of the pixel covered by each corresponding material. This setting gives rise to the linear mixing model, for which several approaches to solving the unmixing problem have been proposed \cite{winter1999n, dobigeon2009joint, bioucas2009variable}. Other works have explored settings in which nonlinear effects, such as reflection and refraction, have a non-negligible effect on the spectral mixtures represented by each pixel. There are a wide range of potential nonlinear unmixing models \cite{heylen2014review, dobigeon2014nonlinear}, and dedicated algorithms for cases where the measured signature is noisy \cite{kendler2019detection,kendler2019non,kendler2022hyperspectral}.

\begin{figure}[t]
    \centering
    \begin{subfigure}{0.45\textwidth}
        \centering
        \includegraphics[width=0.75\linewidth]{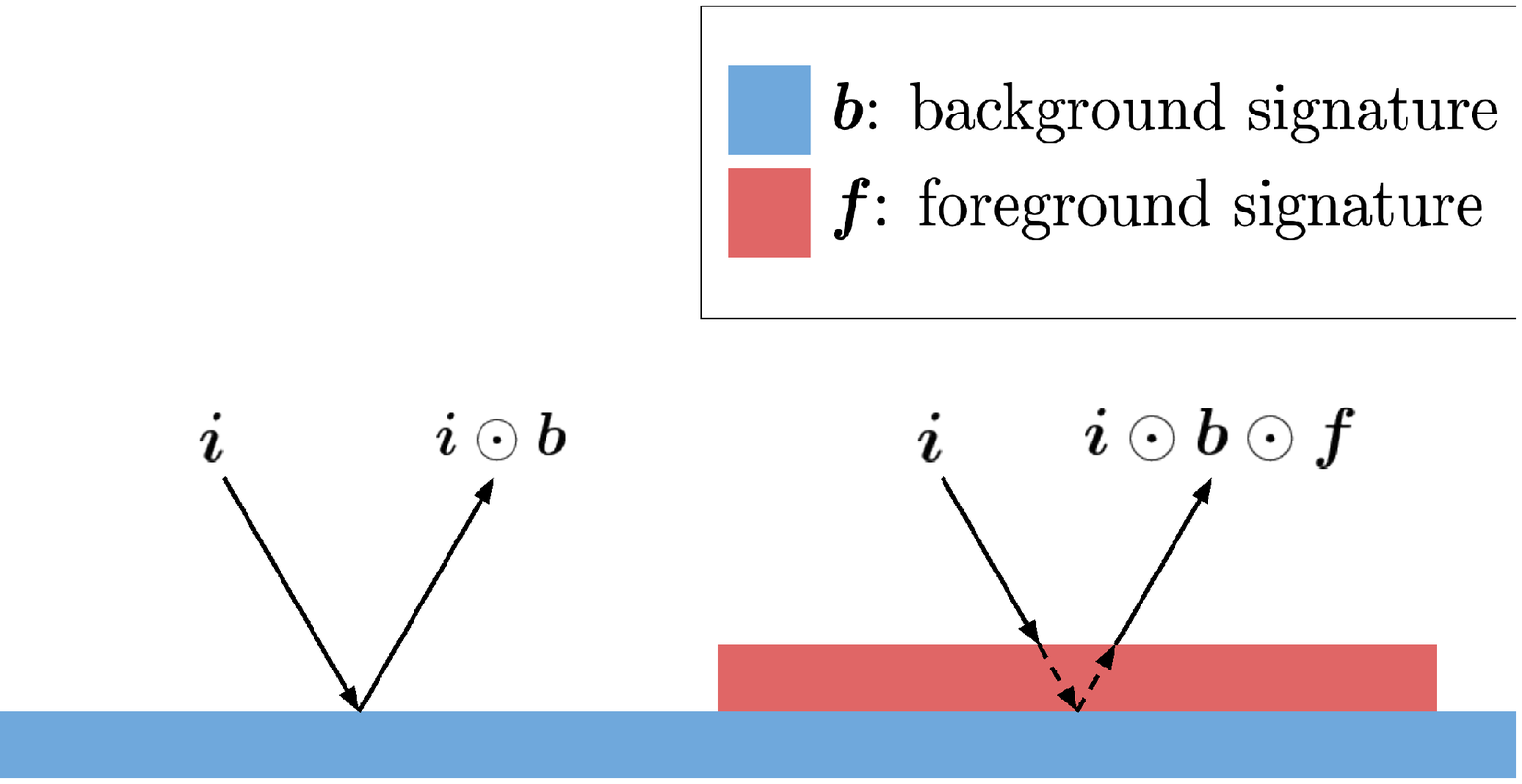}
        \caption{Effect of background only vs background and foreground.}
    \end{subfigure}
    \hfill
    \begin{subfigure}{0.45\textwidth}
        \centering
        \includegraphics[width=0.75\linewidth]{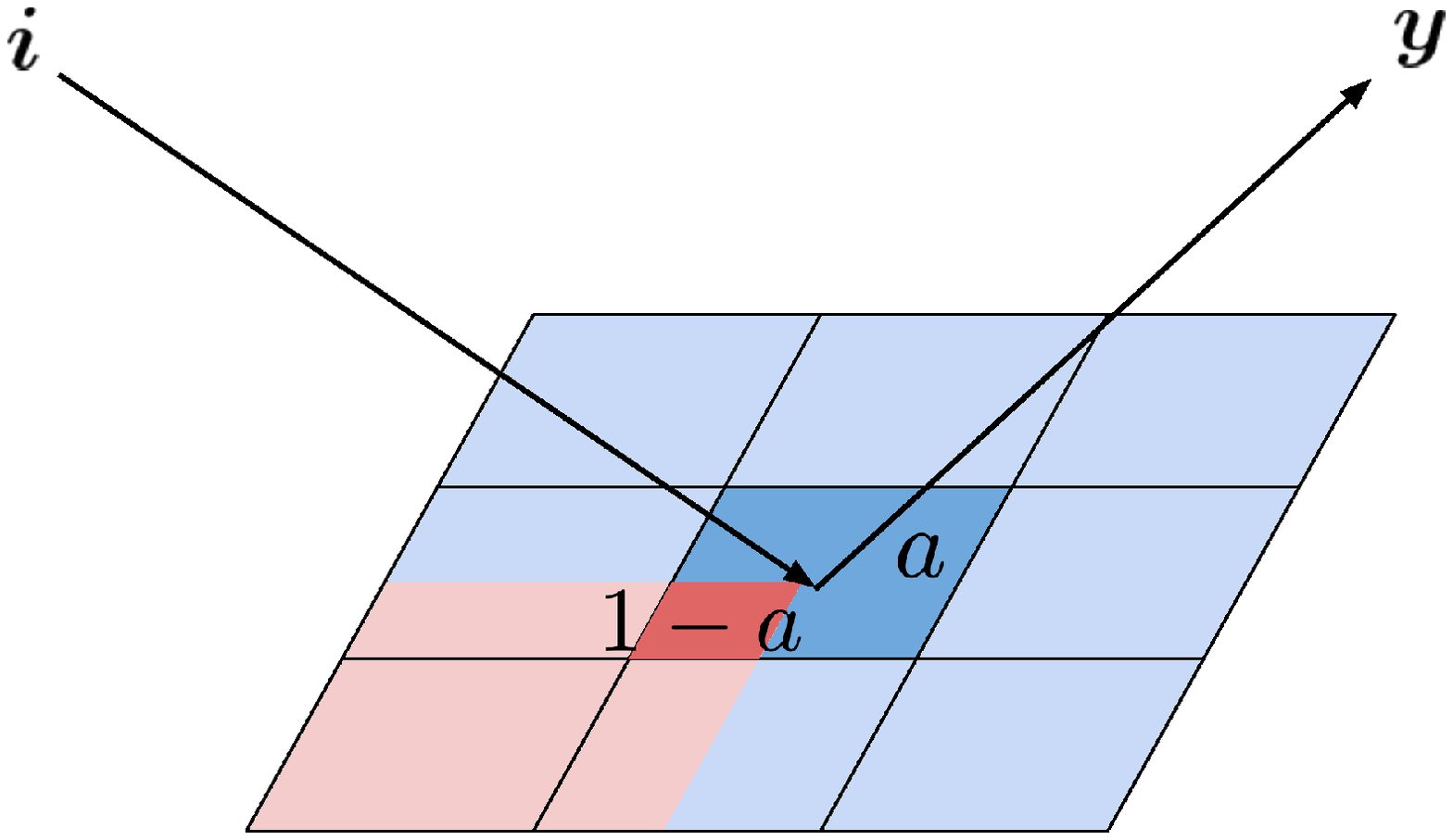}
        \caption{Example of fractional coverage by the foreground material.}
    \end{subfigure}
    
    \caption{Visualizations of the intimate mixing model described in \eqref{eq:kendler_model}. Indexing by $p$ is omitted.}
    
    \label{fig:intimate_mixing_examples}
\end{figure}

In this paper, we consider the non-linear intimate mixing scenario of \cite{kendler2019detection}. In this setting, the measurement of each pixel is modeled as a nonnegative linear combination of the spectral signature from a background material and the product of the spectral signatures from a foreground material and the background material (see Fig.~\ref{fig:intimate_mixing_examples}). The model is expressed as
\begin{equation} \label{eq:kendler_model}
    \vec{y}_p = (\vec{i}_p \odot \vec{b}_p) \odot (a_p \vec{1} + (1-a_p) \vec{f}_p),
\end{equation}
where $\vec{f} \in \mathds{R}^M$ is the foreground material signature, $\vec{b}_p \in \mathds{R}^M$ is the background material signature, $\vec{i}_p \in \mathds{R}^M$ is the reference illumination, and $a_p$ is the foreground material coverage. The notation $\vec{u} \odot \vec{v}$  indicates the Hadamard (element-wise) product between same-sized vectors $\vec{u}$ and $\vec{v}$. As these quantities represent physical phenomena, the material and illumination signatures are strictly positive, and the coverage coefficient satisfies $a_p \in [0, 1]$. In \cite{kendler2019non}, this setting is used to model the detection of suspicious material in envelopes. Our task of interest is to extract the foreground material signature from a set of pixels following the intimate mixing model. This is an unsupervised problem; no a priori information of the material signatures or the material distributions within the set of pixels is assumed. To facilitate this task, we consider some modeling assumptions. We assume only one foreground material is present. Similar to \cite{kendler2019detection} and \cite{lu2013manifold}, we assume that the background material is nearly invariant within small neighborhoods of pixels $\mathcal{N}_k$ (henceforth referred to as \emph{patches}) for $k = 1, \ldots, K$ such that each patch contains exactly one background material. In \cite{kendler2019non}, reference illumination is assumed to be constant; we consider per-pixel scaling of illumination to account for attenuation from physical phenomena such as shadowing. Under these assumptions, the intimate mixing model reduces to our \emph{bag-of-patches} model, and can be expressed for the measurement of pixel $p$ as
\begin{equation} \label{eq:intimate_mixing_model}
    \vec{y}_p = K_p (\vec{i} \odot \vec{b}_k) \odot (a_p \vec{1} + (1 - a_p) \vec{f}), \enspace p \in \mathcal{N}_k,
\end{equation}
where $K_p$ is the non-negative scaling factor for illumination. In practice, patches may be obtained by sampling small regions of pixels such that the assumption of a nearly invariant background material signature holds. Even with these assumptions, the task of extracting the foreground material signature from data following the bag-of-patches model is challenging due to the non-linearity of the model and the redundancy associated with the model parameterization. To illustrate these challenges, we begin by reviewing similar challenges present in the well-known linear mixing model. Later, we will review the bilinear mixing model, which incorporates similar non-linear aspects to the intimate mixing model.

\textbf{Linear mixing model}
Similar to our goal of extracting the foreground material signature in the intimate mixing model, the task of extracting material signatures has been one of the key challenges in the linear mixing model. The measurement model for this setting is
\begin{equation} \label{eq:linear_mixing_model}
    \vec{y}_p = \sum_{r=1}^R \alpha_{r,p} \vec{m}_r,
\end{equation}
where the measurement $\vec{y}_p \in \mathds{R}^M$ of pixel $p$ is expressed as a non-negative linear combination of the material signatures, or \emph{endmembers}, $\vec{m}_1, \vec{m}_2, \ldots, \vec{m}_R \in \mathds{R}^M$, weighted by their \emph{abundances} $\alpha_{1,p}, \alpha_{2,p}, \ldots, \alpha_{R,p}$. Proposed unmixing methods based on a geometric approach are of particular relevance to our work. Geometric approaches can be categorized into pure pixel based methods and volume minimization methods. Pure pixel based methods utilize the \emph{pure pixel} assumption, where the data to be unmixed is assumed to contain at least one pixel for each endmember containing only that endmember. Methods such as the pixel purity test \cite{boardman1993automating}, N-FINDR \cite{winter1999n}, VCA \cite{nascimento2005vertex}, and AVMAX \cite{chan2011simplex} reduce to finding such pure pixels. Volume minimization based methods find a minimum volume simplex that encloses all pixels in the data; the endmembers are the vertices of the obtained simplex. Methods such as MVES \cite{chan2009convex} and MinVolNMF \cite{leplat2019minimum} enforce the enclosure of pixels as a hard constraint, while other methods such as MVSA \cite{li2008minimum} and SISAL \cite{bioucas2009variable} attempt to account for noise by allowing negative abundance estimates with some penalty term. Geometric approaches based on volume minimization do not require the pure pixel assumption to be satisfied, but other data conditions may be necessary.

\textbf{Bilinear mixing model}
Generalizations of the linear mixing model to address nonlinear mixtures of signatures are also considered, such as the bilinear mixing model. The measurement model is
\begin{equation} \label{eq:bilinear_mixing_model}
    \vec{y}_p = \sum_{r=1}^R \alpha_{r,p} \vec{m}_r + \sum_{i=1}^R \sum_{j=1}^R \beta_{i,j,p} \vec{m}_i \odot \vec{m}_j,
\end{equation}
where the measurement $\vec{y}_p \in \mathds{R}^M$ is expressed as a linear combination of the material signatures $\vec{m}_1, \vec{m}_2, \ldots, \vec{m}_R \in \mathds{R}^M$ weighted by the linear abundances $\alpha_{1,p}, \alpha_{2,p}, \ldots, \alpha_{R,p}$, and a linear combination of pairwise products of these material signatures weighted by the bilinear abundances $\beta_{1,1,p}, \beta_{1,2,p}, \ldots, \beta_{R,R,p}$. Several methods have been proposed for solving the bilinear unmixing problem in a \emph{supervised} setting. In \cite{nascimento2009nonlinear}, material signatures are obtained via an oracle (using either label information or expert identification). A material signature matrix is formed with columns containing both the previously identified endmembers and their bilinear combinations, and abundances are then estimated by solving a constrained linear least squares problem as in the linear mixing model (see \cite{heinz2001fully}). In \cite{yokoya2014nonlinear}, it is similarly assumed that endmembers are available directly. Estimates of the abundances for the linear and bilinear combinations are obtained with an alternating minimization of a fitting error objective, alternating between updates for abundances of each type of combination. The iterative update with respect to each parameter has the form of a semi-NMF problem and is solved using existing algorithms \cite{ding2010convex}. In \cite{halimi2011nonlinear}, material signatures are again assumed to be given or estimated via existing algorithms for the linear unmixing problem, such as the pixel purity test or N-FINDR. The abundance coefficients are then obtained via Bayesian estimation, where priors are derived from the constraints on the abundance coefficients and from assumed additive Gaussian random noise in the data. The aforementioned supervised approaches share a common limitation: they are only applicable when the foreground material signature is included in the set of known training signatures. To our knowledge, no works have explored the \emph{unsupervised} unmixing problem in the bilinear mixing model.

Identifiability analysis of the unmixing problem is a fundamental challenge. For the \emph{linear mixing model}, material signatures are said to be identifiable if they can be identified up to some trivial ambiguities (i.e., scaling and permutation) based on the data. There are some known sufficient data conditions that ensure the identifiability of the material signatures. One such condition is \emph{separability} \cite{donoho2004does,laurberg2008theorems}: each material signature must appear in isolation in at least one pixel (referred to as \emph{endmembers}). Such pixels are referred to as endpoints. For a single patch under our proposed model, this would be equivalent to the patch containing a pixel with no foreground material and a pixel entirely covered with foreground material. A more relaxed condition known as sufficient scattering has also been proposed \cite{huang2014non,fu2018identifiability}. \footnote{In the intimate mixing model, the sufficient scattering condition is equivalent to the separability condition.} It has been shown that under the aforementioned data conditions for the linear mixing model, algorithms such as MinVolNMF \cite{leplat2019minimum} can recover the true material signatures. To our knowledge, no works have developed equivalent identifiability conditions for the bilinear mixing model.

In this paper, we focus on \emph{unsupervised} extraction of the foreground signature in the intimate mixing model with identifiability guarantees. We note that the intimate mixing model \eqref{eq:intimate_mixing_model} may be viewed as a special case of the linear mixing model and the bilinear mixing model. However, the existing signature extraction approaches for the linear and bilinear mixing models are not readily applicable to our problem. For instance, the signature extraction methods for the linear mixing model \eqref{eq:linear_mixing_model} could be viable approaches if we are given a patch $\mathcal{N}_k$ that contains two endpoint pixels, i.e., one pixel data being a scalar multiple of $\vec{i} \odot \vec{b}_k$ and another pixel data being a scalar multiple of $\vec{i} \odot \vec{b}_k \odot \vec{f}$ In such a case, the separability condition for the linear mixing model is satisfied by this patch, and thus any signature extraction approach for the linear model can be applied to this patch to identify the two endmembers, $\vec{e}_1 = \vec{i} \odot \vec{b}_k$ and $\vec{e}_2 = \vec{i} \odot \vec{b}_k \odot \vec{f}$, which can be used to obtain $\vec{f}$. However, in practice there is no guarantee that a patch with two endpoint pixels would exist. Furthermore, even when such a patch exists, its identity is unknown and difficult to infer due to the variation of background material signatures across patches. The intimate mixing model \eqref{eq:intimate_mixing_model} can also be seen as a special instance of the bilinear mixing model \eqref{eq:bilinear_mixing_model} after proper reparameterization. However, it does not lead to a viable solution approach due to the lack of unsupervised signature extraction approaches for the bilinear model.  We also note that existing methods for unmixing in the bilinear mixing model do not have identifiability guarantees. Such guarantees are necessary for foreground signature extraction, where the true foreground material signature is potentially unknown and, therefore, not verifiable with an external oracle. The main contributions of this paper are:
\begin{enumerate}
    \item Introduction of a bag-of-patches model for the intimate mixing problem, and characterization of the solution space for the problem;
    \item Development of identifiability data conditions for solutions of the foreground material signature under the bag-of-patches model. We show that, under appropriate data conditions, solutions satisfying the minimum volume and/or endpoint fit properties will match the true foreground material signature up to variations of scaling and element-wise inversion. In contrast to existing identifiability data conditions requiring two endmembers in every patch, the proposed condition requires only two endmembers among \emph{all} patches; and
    \item Proposal of algorithms based on the proposed identifiability criteria to find solutions under this model with identifiability guarantees. One algorithm considers a volume minimization criterion, and the other algorithm is based on an endpoint identification approach.
\end{enumerate}

The remainder of this paper is organized as follows. The observation model and the associated foreground material signature extraction problem are described in Section~\ref{sec:problem_setup}. Conditions under which an identifiable solution for a foreground material signature can be obtained are developed in Section~\ref{sec:theory}. Section~\ref{sec:algorithms} presents our proposed algorithms based on the previously identified conditions. The performance of our proposed methods is evaluated with numerical experiments on both synthetic and real data. The process and results are given in Section~\ref{sec:experiments}, and conclusions and future works are stated in Section~\ref{sec:conclusion}. Proofs of the various theoretical results in the paper are given in the appendix.
 
\section{Problem Setup and Challenges}
\label{sec:problem_setup}

We introduce a summary of notations used in the paper. We then proceed with a formal description of our problem, including the intimate mixing model, our proposed bag-of-patches model, and the ambiguity of representation in this model.

\subsection{Notations}
In this paper, small letters (both roman and greek) are used to denote scalars (e.g., $a$, $\alpha$), boldface small letters (both roman and greek) are used to denote column vectors  (e.g., $\vec{v}$, $\vec{\theta}$), and boldface capital letters are used to denote matrices (e.g., $\vec{M}$). The Hadamard product and division between two vectors is denoted $\vec{u} \odot \vec{v}$ and $\vec{u} \oslash \vec{v}$, respectively, for same sized vectors $\vec{u}$ and $\vec{v}$. The element-wise matrix inequality is denoted $\mat{A} \geq c$ for arbitrarily sized matrix $\mat{A}$ and scalar $c$.

\subsection{Bag-of-Patches Model}
\label{ssec:bag_of_patches_model}

Recall the intimate mixing bag-of-patches model introduced in \eqref{eq:intimate_mixing_model}. This model has many sources of ambiguity: the scaling factors and coverage coefficients are not separable, as well as the reference illumination and each background material signature. The unique identification of these parameters is not relevant to the task of foreground material signature extraction, so we consider the following reparameterization: let $c_{1,p} := K_p a_p$, let $c_{2,p} := K_p (1 - a_p)$, and let $\vec{v}_k := \vec{i} \odot \vec{b}_k$. Note that the restrictions of $K_p$ and $a_p$ (i.e., $K_p \ge 0$ and $a_p \in [0,1]$) require that each $c_{1,p}$ and $c_{2,p}$ also be non-negative. Similarly, the strict positivity of $\vec{i}$ and $\vec{b}_k$ require that $\vec{v}_k$ be strictly positive. The reformulated bag-of-patches model can be expressed as
\begin{equation} \label{eq:bag_of_patches_model}
\begin{gathered}
    \mat{Y}^{(k)} = \diag(\vec{v}^{(k)}) \begin{bmatrix} \vec{f} & \vec{1} \end{bmatrix} \mat{C}^{(k)}, \\
    \vec{v}^{(1)}, \ldots, \vec{v}^{(K)}, \vec{f} > 0, \enspace \mat{C}^{(k)} \geq 0,
\end{gathered}
\end{equation}
where $N_k$ is the number of pixels in each patch, each patch $\mat{Y}^{(k)}$ is an $M \times N_k$ matrix with $\mat{Y}^{(k)}:= \begin{bmatrix} \vec{y}_1 & \cdots & \vec{y}_{N_k} \end{bmatrix}$, $\vec{f}$ and each $\vec{v}^{(k)}$ are $M \times 1$ vectors, and each $\mat{C}^{(k)}$ is a $2 \times N_k$ matrix with $[\mat{C}^{(k)}]_{ij} := c_{i,j}^{(k)}$. We further require that each patch $\mat{Y}^{(k)}$ is rank 2. \footnote{We note that with a rank 1 patch $\mat{Y}^{(k)}$, for any estimate of foreground material signature $\vec{\tilde{f}}$ there exists a pair of vector $\vec{\tilde{v}}^{(k)}$ and matrix $\mat{\tilde{C}}^{(k)}$ such that $\mat{Y}^{(k)} = \diag(\vec{\tilde{v}}^{(k)}) \begin{bmatrix} \vec{\tilde{f}} & \vec{1} \end{bmatrix} \mat{\tilde{C}}^{(k)}$, where $\vec{\tilde{v}}^{(k)} > 0$ and $\mat{\tilde{C}}^{(k)} \geq 0$. Thus, rank 1 patches do not add any additional constraints not introduced by rank 2 patches, and we may ignore the contribution of such patches to the solution.}

\subsection{Problem Formulation}
\label{ssec:problem_formulation}

Given a collection of hyperspectral data $\mat{Y}^{(1)}, \ldots, \mat{Y}^{(K)}$ following the bag-of-patches model in \eqref{eq:bag_of_patches_model} with unknown parameters $\vec{f}$, $\vec{v}^{(k)}$, and $\mat{C}^{(k)}$ for $k = 1, \ldots, K$, our goal is to obtain the foreground material signature $\vec{f}$. To this end, we regard this problem as an estimation problem wherein $\vec{f}$ is the desired parameter and  $\vec{v}^{(k)}$, and $\mat{C}^{(k)}$ for $k = 1, \ldots, K$ are the nuisance parameters.

\noindent {\bf A Key Challenge --- Ambiguity in the Bag-of-Patches Model:}
As previously identified, the factorization in the bag-of-patches model \eqref{eq:bag_of_patches_model} is not unique. Consider the application of a transformation matrix $\mat{T} = \begin{bmatrix} \alpha & \gamma \\ \beta & \delta \end{bmatrix}$ to each patch $\mat{Y}^{(k)}$ in the model:
\begin{equation} \label{eq:transformation_in_model}
\begin{gathered}
\begin{aligned}
    \mat{Y}^{(k)} &= \diag(\vec{v}^{(k)}) \begin{bmatrix} \vec{f} & \vec{1} \end{bmatrix} \mat{T} \mat{T}^{-1} \mat{C}^{(k)} \\
    &= \diag(\vec{\tilde{v}}^{(k)}) \begin{bmatrix} \vec{\tilde{f}} & \vec{1} \end{bmatrix} \mat{\tilde{C}}^{(k)},
\end{aligned} \\
    \vec{\tilde{v}}^{(k)} = (\gamma \vec{f} + \delta \vec{1}) \odot \vec{v}^{(k)}, \\
    \mat{\tilde{C}}^{(k)} = \mat{T}^{-1} \mat{C}^{(k)}, \enspace \vec{\tilde{f}} = (\alpha \vec{f} + \beta \vec{1}) \oslash (\gamma \vec{f} + \delta \vec{1}).
\end{gathered}
\end{equation}
The alternative factorization must still respect the properties of the model: the alternative signatures $\vec{\tilde{f}}$ and $\vec{\tilde{v}}^{(k)}$ must be strictly positive, and the coefficient matrices $\mat{\tilde{C}}^{(k)}$ must be non-negative. These constraints are dependent on the data in the set of patches, and the intersection of these constraints defines the space of admissible transformations in the model. It is clear that the bag-of-patches model may have multiple representations for a given set of hyperspectral data. This is typical of unmixing problems; as previously noted, in the linear mixing model estimates of the endmember and abundance matrices may be obtained up to variations of scaling and permutation \cite{jia2009constrained,fu2018identifiability}. Similarly, we will focus on determining identifiability conditions for the intimate mixing model, and defining the characteristic variations of the set of identifiable solutions under such conditions.

\section{Theory}
\label{sec:theory}

We begin this section with a formal definition for solutions to the foreground material signature extraction problem. Next, we identify the inherent ambiguity of solutions and characterize the space of feasible solutions. Finally, we derive identifiability conditions and associated criterion under which solutions for the foreground material signature will match the true foreground material signature up to the variations of scaling and element-wise inversion.

\subsection{Solution to the Bag-of-Patches Model}

Our goal is to estimate a foreground material signature that satisfies the bag-of-patches model \eqref{eq:bag_of_patches_model} for a given set of patches. We refer to such an estimate as a solution to the bag-of-patches model. For any solution, there must exist associated estimates of the background-illumination signature and coefficient matrix for each patch satisfying the constraints of the bag-of-patches model. We consider these as nuisance parameters, and define a solution only in terms of the estimate of the foreground material signature. The definition of a solution is stated in Definition~\ref{def:solution}.

\begin{definition}[Solution to Bag-of-Patches Model] \label{def:solution}
Let $\mat{Y}^{(1)}, \ldots, \mat{Y}^{(K)}$ be a set of patches that satisfies the bag-of-patches model \eqref{eq:bag_of_patches_model}. We say that a vector $\vec{\tilde{f}} \in \mathds{R}_{++}^M$ is a solution to the bag-of-patches model if there exists $\vec{\tilde{v}}^{(k)} \in \mathds{R}_{++}^M$ and $\mat{\tilde{C}}^{(k)} \in \mathds{R}_+^{2 \times N_k}$ for $k = 1, \ldots, K$ that satisfy
\begin{enumerate}[label={\bf(D\ref{def:solution}:\arabic*)}, ref={(D\ref{def:solution}:\arabic*)}, align=left]
    \item $\mat{Y}^{(k)} = \diag(\vec{\tilde{v}}^{(k)}) \begin{bmatrix} \vec{\tilde{f}} & \vec{1} \end{bmatrix} \mat{\tilde{C}}^{(k)}$; \label{enum:def1_cond1}
    \item $\vec{\tilde{f}} > 0$; \label{enum:def1_cond2}
    \item $\vec{\tilde{v}}^{(k)} > 0, \forall k$; and  \label{enum:def1_cond3}
    \item $\mat{\tilde{C}}^{(k)} \geq 0, \forall k$. \label{enum:def1_cond4}
\end{enumerate}
\end{definition}

\subsection{Solution Ambiguity}

The true foreground material signature $\vec{f}$ is a solution to the bag-of-patches model \eqref{eq:bag_of_patches_model} under Definition~\ref{def:solution}, but other solutions may exist. Section~\ref{ssec:problem_formulation} suggests that some alternative solutions may have the form $\vec{\tilde{f}} = (\alpha \vec{f} + \beta \vec{1}) \oslash (\gamma \vec{f} + \delta \vec{1})$. In fact, the entire space of alternative solutions can be characterized by this form. This result is stated in Property~\ref{prop:space_of_solutions}.

\begin{property} \label{prop:space_of_solutions}
Let $\vec{f} \in \mathds{R}_{++}^M$ be the true foreground material signature for a set of patches $\mat{Y}^{(1)}, \ldots, \mat{Y}^{(K)}$ satisfying the bag-of-patches model \eqref{eq:bag_of_patches_model}. Any solution $\vec{\tilde{f}} \in \mathds{R}_{++}^M$ to the bag-of-patches model satisfies
\begin{equation*}
\begin{aligned}
    \vec{\tilde{f}} &= (\alpha \vec{f} + \beta \vec{1}) \oslash (\gamma \vec{f} + \delta \vec{1}), \\
    \vec{\tilde{v}}^{(k)} &= {\epsilon_k}(\gamma \vec{f} + \delta \vec{1}) \odot \vec{v}^{(k)},~\textrm{for } k=1,2,\ldots,K~ \textrm{and} \\
    \mat{\tilde{C}}^{(k)} &= {\frac{1}{\epsilon_k}} \begin{bmatrix}
        \alpha & \gamma \\
        \beta & \delta
    \end{bmatrix}^{-1} \mat{C}^{(k)}, ~\textrm{for $k = 1,2,\ldots,K$,}
\end{aligned}
\end{equation*}
for some ${\alpha, \beta, \gamma, \delta \in \mathds{R}}$ such that ${\alpha \delta - \beta \gamma \neq 0}$, some ${\epsilon_k > 0}$, some ${\tilde{v}^{(k)} \in \mathds{R}_{++}^M}$, and some ${\tilde{C}^{(k)} \in \mathds{R}_+^{2 \times N_k}}$ for ${k = 1, 2, \ldots, K}$.
\end{property}

The proof of Property~\ref{prop:space_of_solutions} is given in Appendix~\ref{app:pf_space_of_solutions}. Note that for every choice of all strictly positive $\epsilon_k$ for $k = 1, 2, \ldots, K$, there is a corresponding parameterization with reversed sign for the parameters $\alpha$, $\beta$, $\gamma$, and $\delta$, and all strictly negative $\epsilon_k$ for $k = 1, 2, \ldots, K$ that yields identical estimates of the material signatures and coefficient matrices. Thus, we can safely restrict the parameterization of the solution to the case of all positive $\epsilon_k$.

Property~\ref{prop:space_of_solutions} shows that any solution to the bag-of-patches model may be parameterized by coefficients $\alpha,\beta,\gamma,\delta \in \mathds{R}$, $\epsilon_k > 0$ for $k = 1,2,\ldots,K$, and the true model parameters $\vec{f}$ and $\vec{\tilde{v}}^{(k)}, \mat{\tilde{C}}^{(k)}$ for $k = 1,2,\ldots,K$. The estimated illumination-background vectors $\vec{\tilde{v}}^{(k)}$ and coefficient matrices $\mat{\tilde{C}}^{(k)}$ are nuisance parameters, so we seek conditions that characterize a solution only by the previously listed coefficients and the true foreground material signature. Dependence on the true illumination-background vectors can be removed by substituting the definitions of the estimated parameters from Property~\ref{prop:space_of_solutions} in the conditions for a solution given in Definition~\ref{def:solution}. The resulting characterization of a solution is given in the following proposition:

\begin{proposition} \label{prop:feasible_space_solution}
Let $\mat{Y}^{(1)}, \ldots, \mat{Y}^{(K)}$ be a set of patches that satisfies the bag-of-patches model \eqref{eq:bag_of_patches_model} with a true foreground material signature $\vec{f} \in \mathds{R}_{++}^M$ and true coefficient matrices $\mat{C} \in \mathds{R}_+^{2 \times N_k}$ for $k = 1, \ldots, K$. A vector $\vec{\tilde{f}} \in \mathds{R}_{++}^M$ is a solution to the bag-of-patches model according to Definition~\ref{def:solution} if and only if there exist $\alpha, \beta, \gamma, \delta \in \mathds{R}$ such that
\begin{enumerate}[label={\bf(P\ref{prop:feasible_space_solution}:\arabic*)}, ref={(P\ref{prop:feasible_space_solution}:\arabic*)}, align=left]
    \item $\vec{\tilde{f}} = (\alpha \vec{f} + \beta \vec{1}) \oslash (\gamma \vec{f} + \delta \vec{1})$; \label{enum:prop1_cond1}
    \item $\alpha \vec{f} + \beta \vec{1} > 0$; \label{enum:prop1_cond2}
    \item $\gamma \vec{f} + \delta \vec{1} > 0$; \label{enum:prop1_cond3}
    \item $\begin{bmatrix}
        \alpha & \gamma \\
        \beta & \delta
    \end{bmatrix}^{-1} \mat{C}^{(k)} \geq 0, ~\textrm{for $k=1, 2, \ldots, K$}$; and \label{enum:prop1_cond4}
    \item $\alpha \delta - \beta \gamma \neq 0$. \label{enum:prop1_cond5}
\end{enumerate}
\end{proposition}

\begin{proof}
Suppose $\mat{Y}^{(1)}, \ldots, \mat{Y}^{(K)}$ are a set of patches that satisfies the bag-of-patches model \eqref{eq:bag_of_patches_model} with a true foreground material signature $\vec{f} \in \mathds{R}_{++}^M$ and true coefficient matrices $\mat{C}^{(k)} \in \mathds{R}_+^{2 \times N_k}$ for $k = 1, 2, \ldots, K$. Let $\vec{\tilde{f}} \in \mathds{R}_{++}^M$ be a solution to the bag-of-patches model. Condition \ref{enum:prop1_cond1} holds directly from Property~\ref{prop:feasible_space_solution}. Also, from Property~\ref{prop:feasible_space_solution} we have that $\vec{\tilde{v}}^{(k)} = {\epsilon_k}(\gamma \vec{f} + \delta \vec{1}) \odot \vec{v}^{(k)}$ for $k=1,2,\ldots,K$, and $\epsilon > 0$. From Definition~\ref{def:solution}, each $\vec{\tilde{v}}^{(k)}$ is strictly positive. This holds if and only if $\gamma \vec{f} + \delta \vec{1}$ is strictly positive, yielding \ref{enum:prop1_cond3}. Similarly, from Property~\ref{prop:feasible_space_solution} and Definition~\ref{def:solution} we have that $\vec{\tilde{f}} = (\alpha \vec{f} + \beta \vec{1}) \oslash (\gamma \vec{f} + \delta \vec{1})$ and is strictly positive. Using the strict positivity of $\gamma \vec{f} + \delta \vec{1}$, it must hold that $\alpha \vec{f} + \beta \vec{1}$ is strictly positive, yielding \ref{enum:prop1_cond2}. Lastly, substituting the definition of $\mat{\tilde{C}}^{(k)}$ from Property~\ref{prop:space_of_solutions} in \ref{enum:def1_cond4} yields \ref{enum:prop1_cond4}, and \ref{enum:prop1_cond5} holds directly from Property~\ref{prop:feasible_space_solution}.
\end{proof}

Proposition~\ref{prop:feasible_space_solution} connects the notion of a solution to the bag-of-patches model \eqref{eq:bag_of_patches_model} with the space of feasible solutions described in Property~\ref{prop:feasible_space_solution}. Notably, Proposition~\ref{prop:feasible_space_solution} allows a solution to the bag-of-patches model to be described without reference to the true or estimated background-illumination signatures. It is further possible to replace the multiple per-patch constraints in \ref{enum:prop1_cond4} with a single constraint; consider the following lemma:

\begin{lemma}
\label{lem:alternate_coefficient_constraint}
Let $\mat{C}^{(k)} \in \mathds{R}_+^{2 \times N_k}$ for $k = 1, \ldots, K$ be non-negative matrices such that no column is equal to the zero vector and at least one matrix is rank two. Define $r_a$ and $r_b$ as
\begin{equation} \label{eq:r}
    r_a := \min_{j, k} \frac{c_{2,j}^{(k)}}{c_{1,j}^{(k)}} \quad \text{and} \quad r_b = \min_{j, k} \frac{c_{1,j}^{(k)}}{c_{2,j}^{(k)}}.
\end{equation}
For coefficients $\alpha, \beta, \gamma, \delta \in \mathds{R}$ such that $\alpha \delta - \beta \gamma \neq 0$, it holds that
\begin{equation*}
    \begin{bmatrix}
        \alpha & \gamma \\
        \beta & \delta
    \end{bmatrix}^{-1} \mat{C}^{(k)} \geq 0, \enspace \forall k \iff \begin{bmatrix}
        \alpha & \gamma \\
        \beta & \delta
    \end{bmatrix}^{-1} \begin{bmatrix}
        1 & r_b \\
        r_a & 1
    \end{bmatrix} \geq 0.
\end{equation*}
\end{lemma}

\noindent The proof of Lemma~\ref{lem:alternate_coefficient_constraint} is given in Appendix~\ref{app:pf_alternate_coefficient_constraint}. This lemma naturally leads to the following corollary:

\begin{corollary} \label{cor:patch_independent_solution}
Let $\mat{Y}^{(1)}, \ldots, \mat{Y}^{(K)}$ be a set of patches that satisfies the bag-of-patches model \eqref{eq:bag_of_patches_model} with a true foreground material signature $\vec{f} \in \mathds{R}_{++}^M$ and true coefficient matrices $\mat{C}^{(k)} \in \mathds{R}_+^{2 \times N_k}$ for $k = 1, 2, \ldots, K$. Define $r_a$ and $r_b$ as in \eqref{eq:r}. A vector $\vec{\tilde{f}} \in \mathds{R}_{++}^M$ is a solution to the bag-of-patches model if and only if there exist $\alpha, \beta, \gamma, \delta \in \mathds{R}$ such that
\begin{enumerate}[label={\bf(C\ref{cor:patch_independent_solution}:\arabic*)}, ref={(C\ref{cor:patch_independent_solution}:\arabic*)}, align=left]
     \item $\vec{\tilde{f}} = (\alpha \vec{f} + \beta \vec{1}) \oslash (\gamma \vec{f} + \delta \vec{1})$; \label{enum:prop2_cond1}
    \item $\alpha \vec{f} + \beta \vec{1} > 0$; \label{enum:prop2_cond2}
    \item $\gamma \vec{f} + \delta \vec{1} > 0$; and \label{enum:prop2_cond3}
    \item $\begin{bmatrix}
        \alpha & \gamma \\
        \beta & \delta
    \end{bmatrix}^{-1} \begin{bmatrix}
        1 & r_b \\
        r_a & 1
    \end{bmatrix} \geq 0$; and \label{enum:prop2_cond4}
    \item $\alpha \delta - \beta \gamma \neq 0$. \label{enum:prop2_cond5}
\end{enumerate}
\end{corollary}


We have formally defined the set of feasible solutions under the bag-of-patches model in Definition~\ref{def:solution}, and derived a simplified representation of this set in terms of only the foreground material signature in Corollary~\ref{cor:patch_independent_solution}. The set of feasible solutions includes many nonlinear variations of the true foreground material signature, which is difficult to consider for tasks such as foreground material identification or characterization. To address this challenge, we seek conditions and solution criteria under which a set of solutions with a simpler variations may be obtained.

\subsection{Restricting the Solution Space via Identifiability Criteria}

Corollary~\ref{cor:patch_independent_solution} suggests a space of feasible solutions to the bag-of-patches model \eqref{eq:bag_of_patches_model} with many nonlinear variations. We seek a smaller set of feasible solutions with simpler variations, as in identifiable solutions under the linear mixing model. To achieve this, we first explore restricting the solution space by requiring solutions to satisfy additional criteria. We will introduce two potential criteria: the minimum volume criterion, and the endpoint fit criterion. We will show that solutions satisfying either of these criteria are restricted in a manner that will lead to the desired set of solutions if the appropriate identifiability condition is also satisfied.

\vskip \baselineskip
\noindent{\bf Minimum-volume solution: }\quad
The first identifiability criterion we introduce is the minimum volume solution. An estimated foreground material signature $\vec{\tilde{f}}$ satisfies this criterion if it is a local minimum of some volume measure $\vol(\vec{\tilde{f}})$ and satisfies the standard solution constraints. We consider the normalized determinant as our volume measure:
\begin{equation}
    \vol(\vec{\tilde{f}}) = \frac{\det(\begin{bmatrix} \vec{\tilde{f}} & \vec{1} \end{bmatrix}^\trans \begin{bmatrix} \vec{\tilde{f}} & \vec{1} \end{bmatrix})}{\|\vec{1}\|_2^2 \|\vec{\tilde{f}}\|_2^2} = 1 - \left( \frac{(\vec{\tilde{f}}^\trans \vec{1})}{\|\vec{1}\|_2 \|\vec{\tilde{f}}\|_2} \right)^2. \label{eq:volume_measure}
\end{equation}
Using this volume measure, a minimum-volume solution is described in Definition~\ref{def:minimum_volume_solution}.

\begin{definition}[Minimum-Volume Solution] \label{def:minimum_volume_solution}
Let $\mat{Y}^{(1)}, \ldots, \mat{Y}^{(K)}$ be a set of patches that satisfies the bag-of-patches model \eqref{eq:bag_of_patches_model}. A vector $\vec{f^*} \in \mathds{R}_{++}^M$ is a minimum-volume solution to the bag-of-patches model if $\vec{f^*}$ is a local minimizer of \eqref{eq:volume_measure} subject to the constraints in Definition~\ref{def:solution}.
\end{definition}

Minimum-volume solutions to the bag-of-patches model \eqref{eq:bag_of_patches_model} (those that satisfy Definition~\ref{def:minimum_volume_solution}) are well-defined. Such solutions may be stated in terms of a unique element of the feasible space described in Property~\ref{prop:space_of_solutions}, denoted by $\vec{f_0}$ and defined as
\begin{equation} \label{eq:f_0}
    \vec{f_0} = (\vec{f} + r_a \vec{1}) \oslash (r_b \vec{f} + \vec{1}),
\end{equation}
where $r_a$ and $r_b$ are defined as in Lemma~\ref{lem:alternate_coefficient_constraint}. We note that $c \vec{f_0}$ and $c \vec{1} \oslash \vec{f_0}$ for any $c > 0$ are feasible solutions to the bag-of-patches model, as they satisfy the conditions of Corollary~\ref{cor:patch_independent_solution}. \footnote{This may be verified by setting $(\alpha, \beta, \gamma, \delta) = (c, cr_a, r_b, 1)$ for $c \vec{f_0}$ or $(\alpha, \beta, \gamma, \delta) = (cr_b, c, 1, r_a)$ for $c \vec{1} \oslash \vec{f_0}$.} Finally, the nature of minimum-volume solutions is characterized in Theorem~\ref{thm:solutions_to_minvol}.

\begin{theorem} \label{thm:solutions_to_minvol}
Let $\mat{Y}^{(1)}, \ldots, \mat{Y}^{(K)}$ be a set of patches that satisfies the bag-of-patches model \eqref{eq:bag_of_patches_model} with a true foreground material signature $\vec{f} \in \mathds{R}_{++}^M$ and true coefficient matrices $\mat{C}^{(k)} \in \mathds{R}_+^{2 \times N_k}$ for $k = 1, 2, \ldots, K$. Define $\vol(\vec{\tilde{f}})$ as in \eqref{eq:volume_measure}. Let $\vec{f_0}$ be defined as in \eqref{eq:f_0}. A solution $\vec{f^*} \in \mathds{R}_{++}^M$ to the bag-of-patches model is a minimum-volume solution if and only if
\begin{equation*}
    \vec{f^*} = c \vec{f_0} \quad \textrm{or} \quad \vec{f^*} = c \vec{1} \oslash \vec{f_0}, \quad c > 0.
\end{equation*}
\end{theorem}

The proof of Theorem~\ref{thm:solutions_to_minvol} is given in Appendix~\ref{app:pf_solutions_to_minvol}. Note that the variations present in minimum-volume solutions are scaling and element-wise inversion; this is similar to the variations for identifiable solutions to the NMF problem. In particular, element-wise inversion may be viewed as the result of permuting the position of endmember estimates $\vec{\tilde{v}} \odot \vec{\tilde{f}}$ and $\vec{\tilde{v}}$ when computing the ratio to obtain $\vec{\tilde{f}}$. Figure~\ref{fig:point_plot} illustrates an example of a minimum-volume solution. For the example setting in the figure, it is not possible to further reduce the volume between $\vec{v}$ and $\vec{v} \odot \vec{f}$ without violating the non-negativity constraint for patch coefficients. This figure suggests that a minimum-volume solution may result in columns of patches that lie on the vectors $\vec{\tilde{v}}^{(k_1)}$ and $\vec{\tilde{v}}^{(k_2)} \odot \vec{f^*}$ for some patches $\mat{Y}^{(k_1)}$ and $\mat{Y}^{(k_2)}$. We will later see that this is always the case.

\begin{figure}[tbp]
    \centering
    \includegraphics[width=0.9\textwidth]{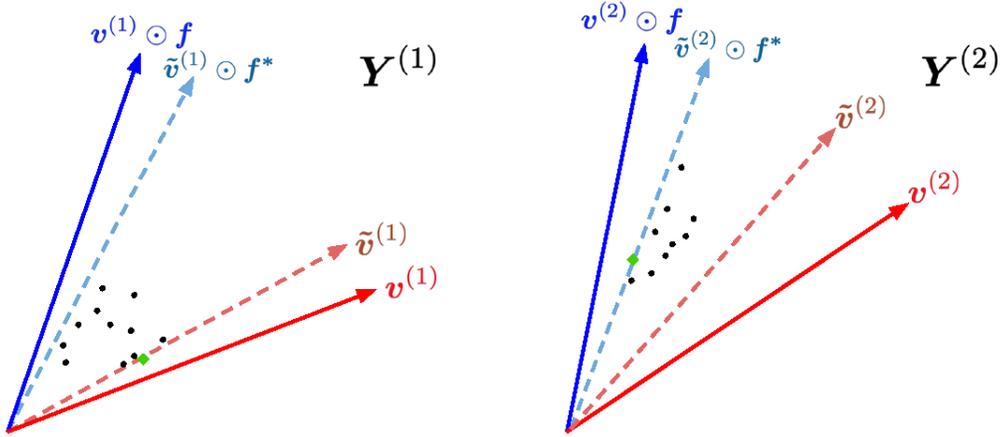}
    \caption{Visualization of a minimum volume (and/or endpoint fit) solution for an example set of two patches $\mat{Y}^{(1)}$ and $\mat{Y}^{(2)}$. The solid blue and red lines correspond to the true $\vec{v}^{(k)} \odot \vec{f}$ and $\vec{v}^{(k)}$ vectors that span each patch. The lighter dashed blue and red lines correspond to alternative $\vec{\tilde{v}}^{(k)} \odot \vec{f^*}$ and $\vec{\tilde{v}}^{(k)}$ vectors. The black dots and green diamonds represent the columns of each patch as they lie in the span of the vector pairs. The green diamonds indicate endpoints.}
    \label{fig:point_plot}
\end{figure}

\vskip \baselineskip
\noindent{\bf End-point solution:}\quad
The second identifiability criterion we introduce is the endpoint fit solution. For a given factorization under the bag-of-patches model, we observe that the columns of a coefficient matrix $\mat{C}^{(k)}$ suggest a notion of coordinates for a given pixel in terms of $\vec{v}^{(k)} \odot \vec{f}$ and $\vec{v}^{(k)}$. The space of valid coordinates is constrained by the non-negativity requirement for coefficients. Consider a coefficient column with one of the following forms: $\begin{bmatrix} x & 0 \end{bmatrix}^\trans$, or $\begin{bmatrix} 0 & y \end{bmatrix}^\trans$, where $x$ and $y$ are positive constants. Any pixel with such a form of coordinates lies at the edge of the coordinate space. We refer to such a pixel as an \emph{endpoint}. A solution $\vec{\tilde{f}}$ satisfies the endpoint fit criterion if, for the corresponding factorization under the bag-of-patches model, there exists at least one of each form of endpoint. The definition follows:

\begin{definition}[Endpoint Fit Solution] \label{def:endpoint_fit_solution}
Let $\mat{Y}^{(1)}, \ldots, \mat{Y}^{(K)}$ be a set of patches that satisfies the bag-of-patches model \eqref{eq:bag_of_patches_model} with a true foreground material signature $\vec{f} \in \mathds{R}_{++}^M$ and true coefficient matrices $\mat{C}^{(k)} \in \mathds{R}_+^{2 \times N_k}$ for $k = 1, 2, \ldots, K$. A solution $\vec{f^*}$ to the bag-of-patches model is an endpoint fit solution if there exist estimated coefficient matrices $\mat{\tilde{C}}^{(k)}$ for $k = 1, 2, \ldots, K$ such that $\mat{\tilde{C}}^{(k)}$ contains a column $\begin{bmatrix} x & 0 \end{bmatrix}^\trans$ for some $k = k_1$, and contains a column $\begin{bmatrix} 0 & y \end{bmatrix}^\trans$ for some $k = k_2$, where $x, y > 0$. \footnote{The indices $k_1$ and $k_2$ need not be distinct. In such a case, both kinds of endpoint occur in the same estimated coefficient matrix.}
\end{definition}

Similar to minimum-volume solutions, endpoint fit solutions to the bag-of-patches model \eqref{eq:bag_of_patches_model} may be stated in terms of $\vec{f_0}$. The nature of endpoint fit solutions is characterized in the following theorem.

\begin{theorem} \label{thm:solutions_to_endfit}
Let $\mat{Y}^{(1)}, \ldots, \mat{Y}^{(K)}$ be a set of patches that satisfies the bag-of-patches model \eqref{eq:bag_of_patches_model} with a true foreground material signature $\vec{f} \in \mathds{R}_{++}^M$ and true coefficient matrices $\mat{C} \in \mathds{R}_+^{2 \times N_k}$ for $k = 1, \ldots, K$. Let $\vec{f_0}$ be defined as in \eqref{eq:f_0}. A solution $\vec{f^*}$ to the bag-of-patches model is an endpoint fit solution if and only if
\begin{equation*}
    \vec{f^*} = c \vec{f_0} \quad \textrm{or} \quad  \vec{f^*} = c \vec{1} \oslash \vec{f_0}, \quad c > 0.
\end{equation*}
\end{theorem}

The proof of Theorem~\ref{thm:solutions_to_endfit} is given in Appendix~\ref{app:pf_solutions_to_endfit}. Figure~\ref{fig:point_plot} can also be used as an example of an endpoint fit solution. This is an endpoint fit because an endpoint in patch $\mat{Y}^{(1)}$ lies on the vector $\vec{\tilde{v}}^{(1)} \odot \vec{f^*}$, and an endpoint in patch $\mat{Y}^{(2)}$ lies on the vector $\vec{\tilde{v}}^{(2)}$. A corollary of Theorem~\ref{thm:solutions_to_minvol} and Theorem~\ref{thm:solutions_to_endfit} is that minimum-volume solutions and endpoint fit solutions are equivalent:

\begin{corollary} \label{cor:equivalence_of_solutions}
Let $\mat{Y}^{(1)}, \ldots, \mat{Y}^{(K)}$ be a set of patches that satisfies the bag-of-patches model \eqref{eq:bag_of_patches_model}. A solution $\vec{f^*} \in \mathds{R}_{++}^M$ to the bag-of-patches model is a minimum-volume solution if and only if it is an endpoint fit solution.
\end{corollary}

\noindent According to Corollary~\ref{cor:equivalence_of_solutions}, every minimum-volume solution is an endpoint fit solution, and vice versa. We will use this result later in developing an algorithm to identify an endpoint fit solution.

\subsection{Complete Solutions with Identifiability Conditions}
\label{ssec:complete_solutions}

The form of minimum-volume solutions and endpoint fit solutions, as described in Theorem~\ref{thm:solutions_to_minvol} and Theorem~\ref{thm:solutions_to_endfit}, may yield the form of solution we have previously considered. Specifically, if $r_a = r_b = 0$ then $\vec{f_0} = \vec{f}$. Hence, any minimum volume or endpoint fit solution $\vec{f^*}$ would be positively proportional to either the true foreground material signature $\vec{f}$ or its element-wise inverse $\vec{1} \oslash \vec{f}$. This data-dependent condition is stated in Definition~\ref{def:fully_tight}.

\begin{definition}[Full-tightness] \label{def:fully_tight}
Let $\mat{Y}^{(1)}, \ldots, \mat{Y}^{(K)}$ be a set of patches that satisfies the bag-of-patches model \eqref{eq:bag_of_patches_model} with a true foreground material signature $\vec{f} \in \mathds{R}_{++}^M$ and true coefficient matrices $\mat{C} \in \mathds{R}_+^{2 \times N_k}$ for $k = 1, \ldots, K$. Suppose there exists a patch $\mat{Y}^{(k_1)}$ containing a scaled version of $\vec{v}^{(k_1)} \odot \vec{f}$ and a patch $\mat{Y}^{(k_2)}$ containing a scaled version of $\vec{v}^{(k_2)}$. Note that we do not require $k_1$ and $k_2$ to be distinct. Equivalently, there exists a column $\begin{bmatrix} \alpha & 0 \end{bmatrix}^\trans$ in $\mat{C^{(k_1)}}$ and a column $\begin{bmatrix} 0 & \beta \end{bmatrix}^\trans$ in $\mat{C^{(k_2)}}$ with $\alpha, \beta > 0$ (this implies $r_a = r_b = 0$). We say that such a set of patches is \emph{fully tight with respect to $\vec{f}$}.
\end{definition}


The full-tightness condition in Definition~\ref{def:fully_tight}, coupled with the minimum volume solution criterion in Definition~\ref{def:minimum_volume_solution} and/or the endpoint fit solution criterion in Definition~\ref{def:endpoint_fit_solution}, ensure that solutions match the true foreground material signature up to variations of scaling and element-wise inversion. We state this result in Theorem~\ref{thm:complete_solution}.

\begin{theorem} \label{thm:complete_solution}
Let $\mat{Y}^{(1)}, \ldots, \mat{Y}^{(K)}$ be a set of patches that satisfies the bag-of-patches model \eqref{eq:bag_of_patches_model} with a true foreground material signature $\vec{f} \in \mathds{R}_{++}^M$ and true coefficient matrices $\mat{C} \in \mathds{R}_+^{2 \times N_k}$ for $k = 1, \ldots, K$. If $\vec{f^*}$ is either a minimum-volume solution or an endpoint fit solution to the bag-of-patches model, and the set of patches is fully tight with respect to $\vec{f}$, then
\begin{equation*}
    \vec{\tilde{f}} = c \vec{f} \enspace\text{or}\enspace \vec{\tilde{f}} = c \vec{1} \oslash \vec{f}, \quad c > 0.
\end{equation*}
\end{theorem}

\begin{proof}
Let $\mat{Y}^{(1)}, \ldots, \mat{Y}^{(K)}$ be a set of patches that satisfies the bag-of-patches model \eqref{eq:bag_of_patches_model} with a true foreground material signature $\vec{f} \in \mathds{R}_{++}^M$ and true coefficient matrices $\mat{C} \in \mathds{R}_+^{2 \times N_k}$ for $k = 1, \ldots, K$. Suppose this set of patches is fully tight with respect to $\vec{f}$, and let $\vec{f_0}$ be defined as in \eqref{eq:f_0}. Then for $r_a$ and $r_b$ as defined in \eqref{eq:r}, it holds that $r_a = r_b = 0$ and therefore $\vec{f_0} = \vec{f}$. If $\vec{f^*}$ is a minimum volume solution or an endpoint fit solution, then by Theorem~\ref{thm:solutions_to_minvol} or Theorem~\ref{thm:solutions_to_endfit} it must hold that $\vec{f^*}$ satisfies
\begin{equation}
    \vec{f^*} = c \vec{f_0} = c \vec{f} \quad \text{or} \quad  \vec{f^*} = c \vec{1} \oslash \vec{f_0} = c \vec{1} \oslash \vec{f}.
\end{equation}
\end{proof}

\noindent In conclusion, any minimum-volume solution or endpoint fit solution will belong to the identifiable set of solutions if the set of patches satisfies the full-tightness condition.

\section{Algorithms}
\label{sec:algorithms}

In the previous section, we proposed two criterion under which identifiable solutions to the foreground material signature extraction problem may be found: a minimum volume solution, and an endpoint fit solution. In this section, we propose two algorithms to solve the foreground material signature extraction problem based on these criterion. First, we consider a projected block coordinate descent algorithm with volume regularization to find solutions satisfying the minimum volume criterion. Then, we adapt the projected block coordinate descent algorithm \emph{without} regularization to instead find solutions satisfying the endpoint fit criterion.

\subsection{Finding Minimum Volume Solutions}

The first approach we consider is to find an estimated foreground material signature that satisfies the minimum volume criterion. According to Definition~\ref{def:minimum_volume_solution}, such a solution must be a local minimum of the volume measure $\vol(\vec{\tilde{f}})$ in \eqref{eq:volume_measure} subject to the constraints \ref{enum:def1_cond1}-\ref{enum:def1_cond4}. The constraints in \ref{enum:def1_cond2} and \ref{enum:def1_cond3} are strict inequalities, which are difficult to consider for optimization. We consider a relaxation of these constraints to be non-strict. It can be shown that any parameters $\vec{\tilde{f}}$ and $\vec{\tilde{v}}^{(k)}$ for $k = 1, 2, \ldots, K$ satisfying the minimum volume criterion under non-strict inequality constraints must be strictly positive. Thus, minimum volume solutions under relaxed constraints are also minimum volume solutions under strict constraints (see Appendix~\ref{app:pf_solutions_to_minvol}).

The exact fitting constraint in \ref{enum:def1_cond1} is also challenging from an optimization perspective. Instead, we take the approach of previous works \cite{zhuang2019regularization,fu2016robust,miao2007endmember} and reformulate the objective as
\begin{equation}
\begin{aligned}
    & g(\vec{\tilde{f}}, \vec{\tilde{v}}^{(1)}, \ldots, \vec{\tilde{v}}^{(K)}, \mat{\tilde{C}}^{(1)}, \ldots, \mat{\tilde{C}}^{(K)}) \\
    &\quad = \sum_{k = 1}^K \|\mat{Y}^{(k)} - \diag(\vec{\tilde{v}}^{(k)}) \begin{bmatrix} \vec{\tilde{f}} & \vec{1} \end{bmatrix} \mat{\tilde{C}}^{(k)}\|_F^2 + \lambda \vol(\vec{\tilde{f}}),
\end{aligned}
\end{equation}
where $\lambda$ is a positive regularization weight that determines the significance of the volume measure in the optimization problem. To obtain a unique solution, we consider an equality constraint for the norms of $\vec{\tilde{f}}$ and $\vec{\tilde{v}}^{(k)}$ for $k = 1, 2, \ldots, K$. Using regularization and the considered constraints, an optimization problem corresponding to a minimum-volume solution is
\begin{equation} \label{eq:vol_min_optimization}
\begin{aligned}
    \text{min} \quad & g(\vec{\tilde{f}}, \vec{\tilde{v}}^{(1)}, \ldots, \vec{\tilde{v}}^{(K)}, \mat{\tilde{C}}^{(1)}, \ldots, \mat{\tilde{C}}^{(K)}) \\
    \text{s.t.} \quad & \vec{\tilde{f}} \geq 0, \norm{\vec{\tilde{f}}} = 1, \\
    & \vec{\tilde{v}}^{(k)} \geq 0, \norm{\vec{\tilde{v}}^{(k)}} = 1, \enspace \text{for $k = 1, 2, \ldots, K$}, \\
    & \mat{\tilde{C}}^{(k)} \geq 0, \enspace \text{for $k = 1, 2, \ldots, K$}.
\end{aligned}
\end{equation}
The minimization problem is separable with respect to each estimated background-illumination signature $\vec{\tilde{v}}^{(k)}$ and estimated coefficient matrix $\mat{\tilde{C}}^{(k)}$ for $k = 1, 2, \ldots, K$. Additionally, the minimization problem for each of these terms depends only on the fitting error of the corresponding patch $\mat{Y}^{(k)}$. Note that the problem is not separable with respect to $\vec{\tilde{f}}$. This suggests an approach based on projected block coordinate descent. An algorithm for this approach is listed in Algorithm~\ref{alg:regularized_fit}. For details on the projection onto the intersection of the non-negative orthant and the surface of the unit sphere, refer to Appendix~\ref{app:projection}. Let $N$ be the number of pixels distributed among all patches; the per-iteration complexity of this algorithm is $O(MN)$.

\begin{algorithm}[t]
\caption{$\mathtt{MinVolFit}$: find a solution to \eqref{eq:vol_min_optimization}.}
\begin{algorithmic} \label{alg:regularized_fit}
    \REQUIRE $\mathcal{Y} = \{\mat{Y}^{(1)}, \ldots, \mat{Y}^{(K)}\}, \lambda, N_\text{iters}$
    \FOR{k = 1, 2, \ldots, K}
        \STATE Randomly initialize $\vec{v}_0^{(k)} \in \mathds{R}_{++}^{M}$ and $\mat{C}_0^{(k)} \in \mathds{R}_{+}^{M \times N_k}$
    \ENDFOR
    \STATE Randomly initialize $\vec{f}_0 \in \mathds{R}_{++}^{M}$
    \STATE Define $f_k(\cdot) = \norm{\mat{Y}^{(k)} - \diag(\vec{v}^{(k)}) \begin{bmatrix} \vec{f} & \vec{1} \end{bmatrix} \mat{C}^{(k)}}_F^2$
    \STATE Define $g(\cdot) = \sum_{k = 1}^K f_k(\cdot) + \lambda \vol(\vec{f})$
    \FOR{$j = 1, \ldots, N_\text{iters}$}
        \STATE In each step, select $\eta$ appropriately via backtracking
        \FOR{k = 1, 2, \ldots, K}
            \STATE $\mat{C}_j^{(k)} \gets P_+ \left( \mat{C}_{j-1}^{(k)} - \eta \left( \left. \nabla_{\mat{C}^{(k)}} f_k(\cdot) \right|_{\vec{v}_{j - 1}^{(k)}, \mat{C}_{j - 1}^{(k)}} \right) \right)$
            \STATE \texttt{\% $P_+$ is the projection onto the non-negative orthant}
            \STATE $\vec{v}_j^{(k)} \gets P_{U+} \left( \vec{v}_{j-1}^{(k)} - \eta \left( \left. \nabla_{\vec{v}^{(k)}} f_k(\cdot) \right|_{\vec{v}_{j - 1}^{(k)}, \mat{C}_{j}^{(k)}} \right) \right)$
            \STATE \texttt{\% $P_{U+}$ is the projection onto the intersection of the non-negative orthant and the surface of the unit sphere}
        \ENDFOR
        \STATE $\vec{f}_j \gets P_{U+} \left( \vec{f}_{j-1} - \eta \left( \left. \nabla_{\vec{f}} g(\cdot) \right|_{\vec{f}_{j - 1}, \vec{v}_{j}^{(k)}, \mat{C}_{j}^{(k)}, \forall k} \right) \right)$
    \ENDFOR
    \RETURN $\vec{f}_{N_\text{iters}}, \vec{v}_{N_\text{iters}}^{(1)}, \ldots, \vec{v}_{N_\text{iters}}^{(K)}, \mat{C}_{N_\text{iters}}^{(1)}, \ldots, \mat{C}_{N_\text{iters}}^{(K)}$
\end{algorithmic}
\end{algorithm}

\subsection{Finding Endpoint Fit Solutions}

Selecting an optimal hyperparameter $\lambda$ for \texttt{RegFit} is not trivial, and experimental results show that the accuracy of the estimated foreground material signature is very sensitive to the choice of hyperparameter. We seek an alternative method that is less sensitive to the choice of hyperparameter. The next algorithmic approach we consider is to find a foreground material signature that satisfies the endpoint fit criterion (see Definition~\ref{def:endpoint_fit_solution}). Consider the following lemma:

\begin{lemma} \label{lem:endpoint_fit}
Let $\mat{Y}^{(1)}, \ldots, \mat{Y}^{(K)}$ be a set of patches generated according to the bag-of-patches model \eqref{eq:bag_of_patches_model}. Suppose $\vec{d}_1, \ldots, \vec{d}_K$ are element-wise positive vectors such that
\begin{equation*}
    \begin{bmatrix}
        \diag(\vec{d}_1) \mat{Y}^{(1)} & \cdots & \diag(\vec{d}_K) \mat{Y}^{(K)}
    \end{bmatrix} = \begin{bmatrix}
        \vec{v}_1 & \vec{v}_2
    \end{bmatrix} \mat{H}.
\end{equation*}
If $\vec{\tilde{y}_i}$ and $\vec{\tilde{y}_j}$ are distinct columns from the matrix $\begin{bmatrix} \diag(\vec{d}_1) \mat{Y}^{(1)} & \cdots & \diag(\vec{d}_K) \mat{Y}^{(K)} \end{bmatrix}$ that minimize
\begin{equation} \label{eq:cosine_similarity}
    \frac{\vec{\tilde{y}_i}^\trans \vec{\tilde{y}_j}}{\norm{\vec{\tilde{y}_i}}_2 \norm{\vec{\tilde{y}_j}}_2},
\end{equation}
then $\vec{f^*} = \vec{\tilde{y}_i} \oslash \vec{\tilde{y}_j}$ is an endpoint fit solution for the set of patches $\mat{Y}^{(1)}, \ldots, \mat{Y}^{(K)}$. 
\end{lemma}

The proof of this lemma is given in Appendix~\ref{app:pf_endpoint_fit}. In the noiseless case, $\mathtt{MinVolFit}$ with $\lambda = 0$ produces a factorization of a set of patches such that the $k$th patch lies in the non-negative span of $\vec{\tilde{v}}^{(k)} \odot \vec{f}$ and $\vec{\tilde{v}}^{(k)}$ for $k = 1, 2, \ldots, K$. Noting that each $\vec{\tilde{v}}^{(k)}$ will be strictly positive (see Appendix~\ref{app:pf_solutions_to_minvol}, we have $\diag(\vec{\tilde{v}}^{(k)})^{-1} \mat{Y}^{(k)} = \begin{bmatrix} \vec{\tilde{f}} & \vec{1} \end{bmatrix} \mat{\tilde{C}}^{(k)}$. Concatenating across all patches for $k = 1, 2, \ldots, K$ yields
\begin{multline}
    \begin{bmatrix}
        \diag(\vec{\tilde{v}}^{(1)})^{-1} \mat{Y}^{(1)} & \cdots & \diag(\vec{\tilde{v}}^{(k)})^{-1} \mat{Y}^{(K)}
    \end{bmatrix} \\
    = \begin{bmatrix}
        \vec{\tilde{f}} & \vec{1}
    \end{bmatrix} \begin{bmatrix} \mat{\tilde{C}}^{(1)} & \cdots & \mat{\tilde{C}}^{(K)}] \end{bmatrix}.
\end{multline}
This follows the form of Lemma~\ref{lem:endpoint_fit}, so an endpoint fit solution is given by the columns of the left matrix that maximize the normalized inner product in \eqref{eq:cosine_similarity}. If patches contain random noise, then each matrix $\diag(\vec{\tilde{v}}^{(k)})^{-1} \mat{Y}^{(k)}$ lies approximately in the non-negative span of $\vec{\tilde{f}}$ and $\vec{1}$, and the endpoint fit solution suggested by Lemma~\ref{lem:endpoint_fit} is an approximate solution. A procedure for finding an endpoint fit solution to the bag-of-patches model using this principle and Algorithm~\ref{alg:regularized_fit} is stated in Algorithm~\ref{alg:endpoint_fit}.

\begin{algorithm}[t]
\caption{$\mathtt{EPFit}$: find a solution to the bag-of-patches model \eqref{eq:bag_of_patches_model} satisfying the endpoint fit criterion.}
\begin{algorithmic} \label{alg:endpoint_fit}
    \REQUIRE $\mathcal{Y} = \{\mat{Y}^{(1)}, \ldots, \mat{Y}^{(K)}\}$
    \STATE $\vec{\tilde{v}}^{(1)}, \ldots, \vec{\tilde{v}}^{(K)} \gets \mathtt{MinVolFit}(\mathcal{Y}, \lambda = 0)$
    \STATE $\mat{\tilde{Y}} \gets \begin{bmatrix} \diag\left(\vec{\tilde{v}}^{(1)}\right)^{-1} \mat{Y}^{(1)} & \cdots & \diag\left(\vec{\tilde{v}}^{(K)}\right)^{-1} \mat{Y}^{(K)} \end{bmatrix}$
    \STATE $\vec{u_1}, \vec{u_2} \gets \vec{\tilde{y}_1}, \vec{\tilde{y}_2}$
    \FOR{$n = 3, 4, \ldots, N$}
        \STATE $\vec{w} \gets \vec{\tilde{y}_n}$
        \IF{$\vec{u_1}^\trans \vec{w} / (\norm{\vec{u_1}}_2 \norm{\vec{w}}_2) > \vec{u_1}^\trans \vec{u_2} / (\norm{\vec{u_1}}_2 \norm{\vec{u_2}}_2)$}
            \STATE $\vec{u_2} \gets \vec{w}$
        \ELSIF{$\vec{w}^\trans \vec{u_2} / (\norm{\vec{w}}_2 \norm{\vec{u_2}}_2) > \vec{u_1}^\trans \vec{u_2} / (\norm{\vec{u_1}}_2 \norm{\vec{u_2}}_2)$}
            \STATE $\vec{u_1} \gets \vec{w}$
        \ENDIF
    \ENDFOR
    \STATE $\vec{f^*} \gets \vec{u_1} \oslash \vec{u_2}$
    \RETURN $\vec{f^*}$
\end{algorithmic}
\end{algorithm}

\section{Experiments}
\label{sec:experiments}

Our experiments are intended to provide empirical verification of our algorithms in comparison to benchmark approaches in a variety of settings. We consider both synthetic and real data experiments. Synthetic data experiments demonstrate the effectiveness of our algorithms in response to particular choices of SNR and data distribution. Real data experiments demonstrate the robustness of our algorithms to the unknown variations in a practical scenario.

\subsection{Synthetic Data Experiments}
\label{ssec:synth_experiments}

Our goals for synthetic data experiments are as follows: to verify that our proposed algorithms can identify foreground material signatures with lower error than benchmark approaches; to demonstrate the effectiveness of each proposed algorithm as a function of SNR; and to show the effect of different data modeling assumptions on the accuracy of all methods. We will show that our algorithms, which can adapt to varying backgrounds per patch, will obtain more accurate foreground material signature estimates than a benchmark algorithm which does not account for varying backgrounds. We will show that the $\mathtt{EPFit}$ algorithm will show an inverse relationship between SNR and the angular difference between the estimated and true foreground material signatures. For the $\mathtt{MinVolFit}$ algorithm, we will show that for each value of SNR there exists an appropriate choice of regularization weight that will minimize the angular difference. Finally, we will show that the $\mathtt{MinVolFit}$ algorithm will show a higher resilience to noise than the $\mathtt{EPFit}$ algorithm.

\noindent {\bf Data generation:} To facilitate synthetic data experiments, we create $K$ patches consisting of $N$ pixels per patch and with $M$ entries per pixel, with pixels lying in the cone between a background-illumination vector and the foreground-background product vector. We generate a single foreground material signature shared among all patches, and distinct background-illumination vectors for each patch. The background-illumination vectors are formed by summing a shared component and weighted individual component with weight $r$. Patches are randomly assigned as \emph{tight} or \emph{non-tight} with probability $p$. We consider two distinct patch settings for our synthetic data experiments: a setting with strictly tight patches, and a setting with only partially tight patches. In the \emph{strictly tight} patch setting, tight patches are generated such that at least one pixel is a scaled version of the background-illumination vector, and another pixel is a scaled version of the product vector. In the \emph{partially tight} patch setting, tight patches are generated such that one pixel is a scaled version of either the background-illumination vector or the product vector, and not both. Finally, the pixels are perturbed by i.i.d. additive Gaussian noise with variance $\sigma^2$. The parameter $\mathtt{is\_strict}$ indicates whether patches may be generated as strictly-tight or partially-tight.
The output of the data generation procedure are the patches $\mat{Y}^{(k)}$, ground truth background-illumination vectors $\vec{v}^{(k)}$, and ground truth coefficient matrices $\mat{C}^{(k)}$ for $k = 1, \ldots, K$, and the ground truth foreground material signature $\vec{f}$.

\noindent {\bf Algorithms:} For our experiments, we consider a benchmark method and our two proposed methods. The benchmark method is a variation of the volume-minimizing NMF method proposed in \cite{fu2019nonnegative}. Given a non-negative input matrix $\mat{Y} \in \mathds{R}_+^{M \times N}$, the benchmark method solves the following problem:
\begin{equation} \label{eq:minvol_nnmf}
\begin{aligned}
    \text{min} \quad & \norm{\mat{Y} - \mat{W} \mat{H}}_F^2 + \lambda \operatorname{logdet} \left( \mat{W}^\trans \mat{W} + \delta \mat{I} \right) \\
    \text{s.t.} \quad & \mat{W} \in \mathds{R}_+^{M \times 2}, \enspace \mat{H} \in \mathds{R}_+^{2 \times N},
\end{aligned}
\end{equation}
where the $\delta \mat{I}$ term for small but sufficiently large $\delta > 0$ ensures that the determinant is positive. To solve this problem, we use the $\mathtt{MinVolNMF}$ algorithm provided in \cite{leplat2019minimum}. Given a patch following the bag-of-patches model, the NMF problem is equivalent to finding matrices $\mat{W} = \begin{bmatrix} \vec{v} \odot \vec{f} & \vec{v} \end{bmatrix}$ and $\mat{H} = \mat{C}$. The foreground material signature for the patch can be extracted by taking the element-wise ratio of columns of $\mat{W}$. Note that volume-minimizing NMF methods allow for permutation and scaling of the recovered columns of $\mat{W}$. If the correct columns of $\mat{W}$ are identified, then the extracted foreground material signature $\vec{\tilde{f}}$ for the patch may be of the form $\vec{\tilde{f}} \propto \vec{f}$ or $\vec{\tilde{f}} \propto \vec{1} \oslash \vec{f}$. We adapt this benchmark method to the multiple patch setting by concatenating each patch $\mat{Y}^{(k)}$ for $k = 1, \ldots, K$ such that $\mat{Y}_\text{all} = \begin{bmatrix} \mat{Y}^{(1)} & \cdots & \mat{Y}^{(K)} \end{bmatrix}$. Then, we solve \eqref{eq:minvol_nnmf} for the concatenated matrix. See Algorithm~\ref{alg:benchmark_1} for reference.

\begin{algorithm}[t]
\caption{Adaptation of $\mathtt{MinVolNMF}$: foreground signature extraction via concatenated patches.}
\begin{algorithmic} \label{alg:benchmark_1}
    \REQUIRE $\mat{Y}^{(1)}, \ldots, \mat{Y}^{(K)}, \delta, \lambda$
    \STATE $\mat{Y}_\text{all} \gets \begin{bmatrix} \mat{Y}^{(1)} & \cdots & \mat{Y}^{(K)} \end{bmatrix}$
    \STATE $\mat{W} \gets \text{MinVolNMF}(\mat{Y}_\text{all},\lambda)$
    \STATE $W_{ij} \gets \delta$ for all $i, j$ s.t. $W_{ij} < \delta$
    \STATE $\vec{f^*} \gets \vec{W}_{:,1} \oslash \vec{W}_{:,2}$
    \RETURN $\vec{f^*}$
\end{algorithmic}
\end{algorithm}

For our proposed methods, we have several hyperparameters to select. For the $\mathtt{MinVolFit}$ algorithm, we must select a regularization weight and an iteration limit. After initial testing, we selected $N_\text{iters} = 1 \times 10^6$; we explore the impact of various choices of regularization weight. For the $\mathtt{EPFit}$ algorithm, we must select an iteration limit. By similar approach, we selected $N_\text{iters} = 5 \times 10^4$.

\noindent {\bf Evaluation metric:} We determine the accuracy of estimates for the foreground material signature by measuring the angular difference between the estimated and true foreground material signatures. The angular difference between two vectors $\vec{u}$ and $\vec{v}$ in degrees is
\begin{equation}
    \Delta \theta = \frac{180}{\pi} \cdot \cos^{-1}\left( \frac{\vec{u}^\trans \vec{v}}{\norm{\vec{u}} \norm{\vec{v}}} \right).
\end{equation}
The angular difference shares a particular relation to the normalized MSE measure; the angular difference can be computed from normalized MSE using
\begin{equation} \label{eq:angular_difference_from_mse}
    \Delta \theta = 2 \sin^{-1}(\sqrt{\text{nMSE}(\vec{u}, \vec{v})}/2) \times \frac{180}{\pi},
\end{equation}
where nMSE computed between $\vec{u}$ and $\vec{v}$ is
\begin{equation}
    \text{nMSE}(\vec{u}, \vec{v}) := \norm{\frac{\vec{u}}{\norm{\vec{u}}} - \frac{\vec{v}}{\norm{\vec{v}}}}^2.
\end{equation}
Our proposed methods, and the benchmark method, can only provide an estimate $\vec{\tilde{f}}$ of the true foreground material signature $\vec{f}$ up to the variations of scaling and elementwise inverse. To accommodate these variations, we consider both the normal and elementwise inverse forms of the identified solution as candidates for comparison. The nMSE for the standard and inverse forms are given by $\text{nMSE}(\vec{f}, \vec{\tilde{f}})$ and $\text{nMSE}(\vec{f}, \vec{1} \oslash \vec{\tilde{f}})$, respectively. We take the minimum of the two measures as the value of nMSE in \eqref{eq:angular_difference_from_mse} and compute the angular difference.

\noindent {\bf Results and analysis:} For our experiments, we explore the relation between SNR and the angular difference between estimated foreground material signatures and true signatures for all considered algorithms over a range of data generation schemes. We conduct three experiments in a strictly tight patch setting, and three experiments in a partially tight patch setting. For each patch setting, we conduct one experiment each with the following ratios of expected magnitudes for the base and varying background-illumination signature components: $r = 0.1$, $r = 0.2$, and $r = 1$. For each experiment, we consider 10 SNR values logarithmically spaced over the range of $10^2$ to $10^6$. Each value of SNR has a corresponding value of noise variance in the data generation procedure described previously; this value is computed as the ratio of the average squared magnitude of all pixel entries in the data and SNR. For each noise level, we generate a set of 20 bags consisting of 10 patches each as detailed in the data generation section above, with each patch containing 25 pixels of dimension 30. In all cases we use $p = 0.5$. To each bag of patches, we apply: $\mathtt{MinVolFit}$ with seven values of $\lambda$ logarithmically spaced between $1 \times 10^{-5}$ and $1 \times 10^{-3}$, $\mathtt{EPFit}$, and adapted $\mathtt{MinVolNMF}$ with seven values of $\lambda$ logarithmically spaced between $1 \times 10^{-2}$ and $1 \times 10^{0}$ and with the default $\delta = 0.1$. We then compute the angular difference between the obtained estimates of the foreground material signature and the true foreground material signature as described in the evaluation metric section above. We report the median angular difference as a function of SNR for each algorithm.

\begin{figure*}[htbp]
    \centering
    \includegraphics[width=0.9\textwidth]{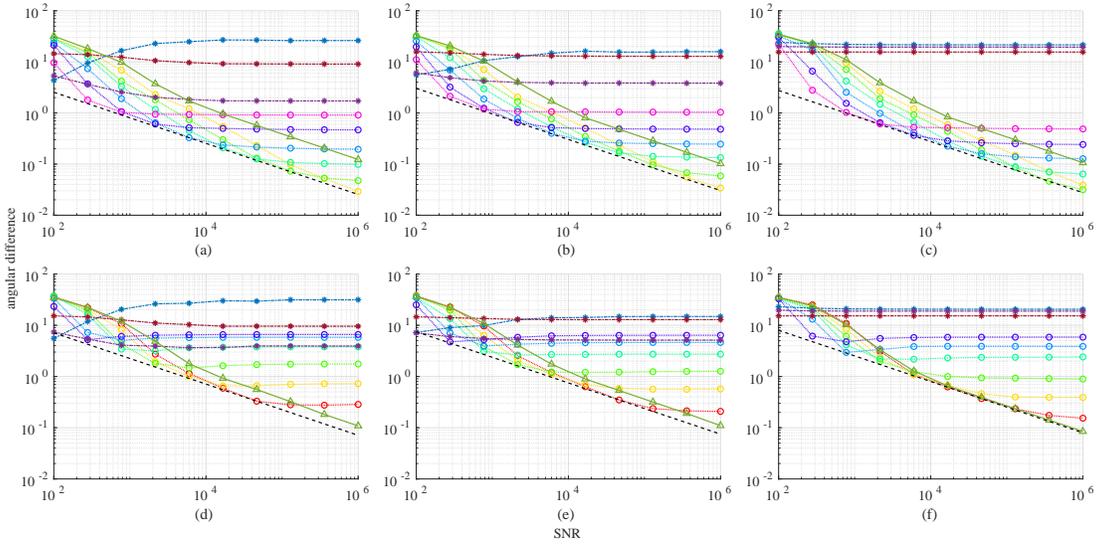}
    \caption{Plot of median angular difference (measured in degrees) between the estimated and true foreground signatures as a function of SNR for the following algorithms: $\mathtt{MinVolNMF}$ (dash-dot line, *) w/ $\lambda = 1 \times 10^{-2}$ (blue), $\lambda = 1 \times 10^{-1}$ (red), and $\lambda = 1 \times 10^{0}$ (purple), $\mathtt{MinVolFit}$ (dotted line, circle) w/ $\lambda = 1 \times 10^{-5}$ (red), $\lambda = 2.154 \times 10^{-5}$ (yellow), $\lambda = 4.642 \times 10^{-5}$ (green), $\lambda = 1 \times 10^{-4}$ (cyan), $\lambda = 2.154 \times 10^{-4}$ (blue), $\lambda = 4.642 \times 10^{-4}$ (purple), and $\lambda = 1 \times 10^{-3}$ (magenta), and $\mathtt{EPFit}$ (olive green, solid line, triangle). A reference slope for angular difference vs SNR matching the lower envelope of $\mathtt{MinVolFit}$ is shown (black dashed line). Each algorithm is applied to a collection of ten patches containing twenty-five pixels each, generated according to the strictly-tight patch setting (a)-(c) or the partially-tight patch setting (d)-(f), and with parameter $r = 0.1$ (a) and (d), $r = 0.2$ (b) and (e), or $r = 1$ (c) and (f). See Section~\ref{ssec:synth_experiments}: data generation for further details. Plots for \texttt{MinVolNMF} w/ $\lambda = 1 \times 10^{-5}$ (red) are omitted from (a)-(c), and w/ $\lambda = 1 \times 10^{-3}$ (magenta) from (d)-(f), as these curves do not intersect the lower envelope line within the selected range of SNR.}
    \label{fig:synth_results}
\end{figure*}

Figure~\ref{fig:synth_results} illustrates the performance of the proposed $\mathtt{MinVolFit}$ and $\mathtt{EPFit}$ algorithms, and the adapted $\mathtt{MinVolNMF}$ benchmark, for both the strictly tight (top row) and partially tight (bottom row) patch settings. For the $\mathtt{MinVolFit}$ algorithm, each choice of regularization weight yields a performance curve with decreasing angular difference as SNR increases up to some limit in SNR, after which the angular difference no longer decreases. As the regularization weight decreases, the SNR value at which the minimum angular difference is achieved increases. The lower envelope of performance curves among all choices of regularization weight (depicted by the black line) shows a consistent inverse relationship between SNR and angular difference. This may be explained by considering the effect of different choices of regularization weight. In a noisy setting, the noise on endpoints may cause the optimal fit of the foreground material signature with respect to fitting error to have a larger volume relative to the true foreground material signature. Volume regularization allows for an increase in fitting error by rewarding a decrease in volume, so for a given value of SNR there should exist an optimal choice of regularization weight that sufficiently reduces the volume of the estimated foreground material signature. As SNR increases, the effect of noisy endpoints is reduced and a smaller regularization weight will be optimal.

For the $\mathtt{EPFit}$ algorithm, in all plots we observe the angular difference decreasing as SNR increases. In contrast to both proposed methods, the benchmark algorithm demonstrates minimal decrease in angular difference as SNR increases. This may be explained by considering the effect of a per-patch varying background material signature. The benchmark approach finds a rank-2 NMF representation of the concatenated set of patches, but the varying background material signatures act as a secondary source of noise when considering the concatenated set of patches. Thus, for sufficiently high SNR the effect of the varying background on benchmark performance is more significant than the effect of additive noise, leading to a lower bound for angular difference. The proposed algorithms account for varying background material signatures among patches, so we do not observe the same issue in these methods. This explanation is further supported by the effect of changing the ratio of expected magnitudes of background material signature components. As we increase the ratio (from subfigures (a) to (c) and from subfigures (d) to (f)), the background material signatures become more varied between patches. Subsequently, the minimum angular difference observed for the benchmark algorithm increases. Again, the proposed algorithms are designed to allow for varying background material signatures, so we do not see a notable impact on the performance of these methods.

In comparing the performance of the proposed methods, we observe that for a given value of SNR, the $\mathtt{MinVolFit}$ algorithm achieves a smaller median angular difference than the $\mathtt{EPFit}$ for at least one choice of hyperparameter. This is reasonable: the $\mathtt{EndpointFit}$ is sensitive to noise affecting endpoints. In contrast, the $\mathtt{MinVolFit}$ algorithm can accommodate some amount of noise via careful selection of the regularization weight $\tau$. Any non-zero value of $\tau$ allows a slight increase in fitting error if it allows a decrease in the volume measure, which can allow the estimated foreground material signature to have a smaller volume measure than would be given by noisy endpoints. However, selecting an optimal regularization weight is non-trivial. In contrast, the $\mathtt{EPFit}$ algorithm does not require any hyperparameter selection.

\subsection{Real Data Experiments}

To verify that our algorithms can perform in a practical setting with potentially unknown variations, we consider experiments in a real data scenario. Our goals for real data experiments are as follows: to demonstrate that algorithms which account for per-patch background variation can perform better than algorithms which do not take such variation into account; and to verify that our algorithms have some robustness to unknown variations that may be present in real data. We will show that our algorithms, which can adapt to varying backgrounds per patch, will obtain more accurate foreground material signature estimates than a benchmark algorithm which does not account for varying backgrounds. For the $\mathtt{MinVolFit}$ algorithm, we will show that for an appropriate choice of regularization weight the estimated foreground material signature will be close to the expected signature. Similarly, for the $\mathtt{EPFit}$ algorithm we will show that the estimated foreground material signature will be close to the expected signature.

\noindent {\bf Data sampling:}
We use the dataset created by Kendler {\em et al.}~\cite{kendler2019detection}. The dataset consists of several annotated hyperspectral cubes of a particular scene. The scene consists of many background materials, with a mix of five distinct foreground materials (sugar, polystyrene, silicone, white silicone, and jam) deposited on the surfaces of several background materials. The locations of background and foreground materials in the scene are annotated in the dataset. However, the specific coverage at the per-pixel level is unknown; the annotation specifies only whether a pixel contains any amount foreground material. To facilitate experiments, we selected regions containing two different background materials (ceramic tile and plywood) and one shared foreground material (silicone deposit). We sampled patches by sweeping a $12 \times 12$ square window with a one-pixel offset through each region.

\noindent {\bf Algorithms:}
To account for the significant noise present in real data, we make slight modifications to the $\mathtt{EPFit}$ algorithm. In this, the matrix $\mat{\tilde{Y}}$ contains the vectors of all patches projected onto a rank-2 span. The final step of the algorithm is to find the pair of vectors with maximum angular difference, under the assumption that these vectors represent endpoints in the original data. In the presence of noise, it is possible for vectors which are not endpoints to be perturbed such that they produce a larger angular difference with respect to other vectors in the span than the true endpoints. To account for these possible errors, we consider removing some columns from $\mat{\tilde{Y}}$ which yield the largest angular difference. This procedure is given in Algorithm~\ref{alg:noisy_endpoint_removal}.

\begin{algorithm}[t]
\caption{Endpoint identification method with noisy endpoint removal.}
\begin{algorithmic} \label{alg:noisy_endpoint_removal}
    \REQUIRE $\mat{\tilde{Y}}, \alpha$
    \STATE $\mat{V} \gets \text{rank-2 basis of $\mat{\tilde{Y}}$}$
    \STATE $\mat{C} \gets \mat{V}^+ \mat{\tilde{Y}}$
    \STATE $\vec{r} \gets \mat{C}_{1,:} \oslash \mat{C}_{2,:}$
    \STATE remove columns from $\mat{\tilde{Y}}$ with the $\alpha$ smallest and largest coefficient ratios in $\vec{r}$
    \STATE $\vec{u_1}, \vec{u_2} \gets \vec{\tilde{y}_1}, \vec{\tilde{y}_2}$
    \FOR{$n = 3, 4, \ldots, N$}
        \STATE $\vec{w} \gets \vec{\tilde{y}_n}$
        \IF{$\vec{u_1}^\trans \vec{w} / (\norm{\vec{u_1}}_2 \norm{\vec{w}}_2) > \vec{u_1}^\trans \vec{u_2} / (\norm{\vec{u_1}}_2 \norm{\vec{u_2}}_2)$}
            \STATE $\vec{u_2} \gets \vec{w}$
        \ELSIF{$\vec{w}^\trans \vec{u_2} / (\norm{\vec{w}}_2 \norm{\vec{u_2}}_2) > \vec{u_1}^\trans \vec{u_2} / (\norm{\vec{u_1}}_2 \norm{\vec{u_2}}_2)$}
            \STATE $\vec{u_1} \gets \vec{w}$
        \ENDIF
    \ENDFOR
    \STATE $\vec{f^*} \gets \vec{u_1} \oslash \vec{u_2}$
    \RETURN $\vec{f^*}$
\end{algorithmic}
\end{algorithm}

\noindent {\bf Evaluation scheme:}
Unlike the synthetic data experiments, the true foreground material signatures for the real data experiments are unknown. However, with the label information for the locations of foreground and background materials, we can make a reasonable inference as to the value of the true foreground material signature up to scaling and noise effects. Given a region containing exactly one background and foreground material, we may compare all possible foreground pixels with all possible background pixels and extract candidate foreground signatures by taking the elementwise ratio of these combinations. To avoid issues of small variations in the background material, we select pairs of pixels that are within 10 pixels of each other. This is similar to the background subtraction method in \cite{kendler2019detection}. Then, we may sort all candidate foreground material signatures by their volume measure (see \eqref{eq:volume_measure}) and obtain a denoised reference foreground material signature by averaging the top $k$ signatures. For our experiments, we let $k = 10$. See Algorithm~\ref{alg:oracle} for reference.

\begin{algorithm}[t]
\caption{Oracle: extract reference foreground material signature for real data experiments using label information.}
\begin{algorithmic} \label{alg:oracle}
    \REQUIRE $\mat{X}, \vec{y}, \delta, K$
    \STATE $\mathcal{F} \gets \emptyset$
    \FOR{all $i$ s.t. $y_i = \text{foreground}$}
        \FOR{all $j$ s.t. $y_j = \text{background}$}
            \IF{distance between $\vec{X}_i$ and $\vec{X}_j$ is less than $\delta$}
                \STATE $\mathcal{F} \gets \mathcal{F} \cup \{\vec{X}_i \oslash \vec{X}_j\}$
            \ENDIF
        \ENDFOR
    \ENDFOR
    \STATE $\vec{f^*} \gets \text{mean of top $K$ elements of $\mathcal{F}$ w.r.t. $\vol$}$
    \RETURN $\vec{f^*}$
\end{algorithmic}
\end{algorithm}

\begin{figure*}[htbp]
    \centering
    \includegraphics[width=0.9\textwidth]{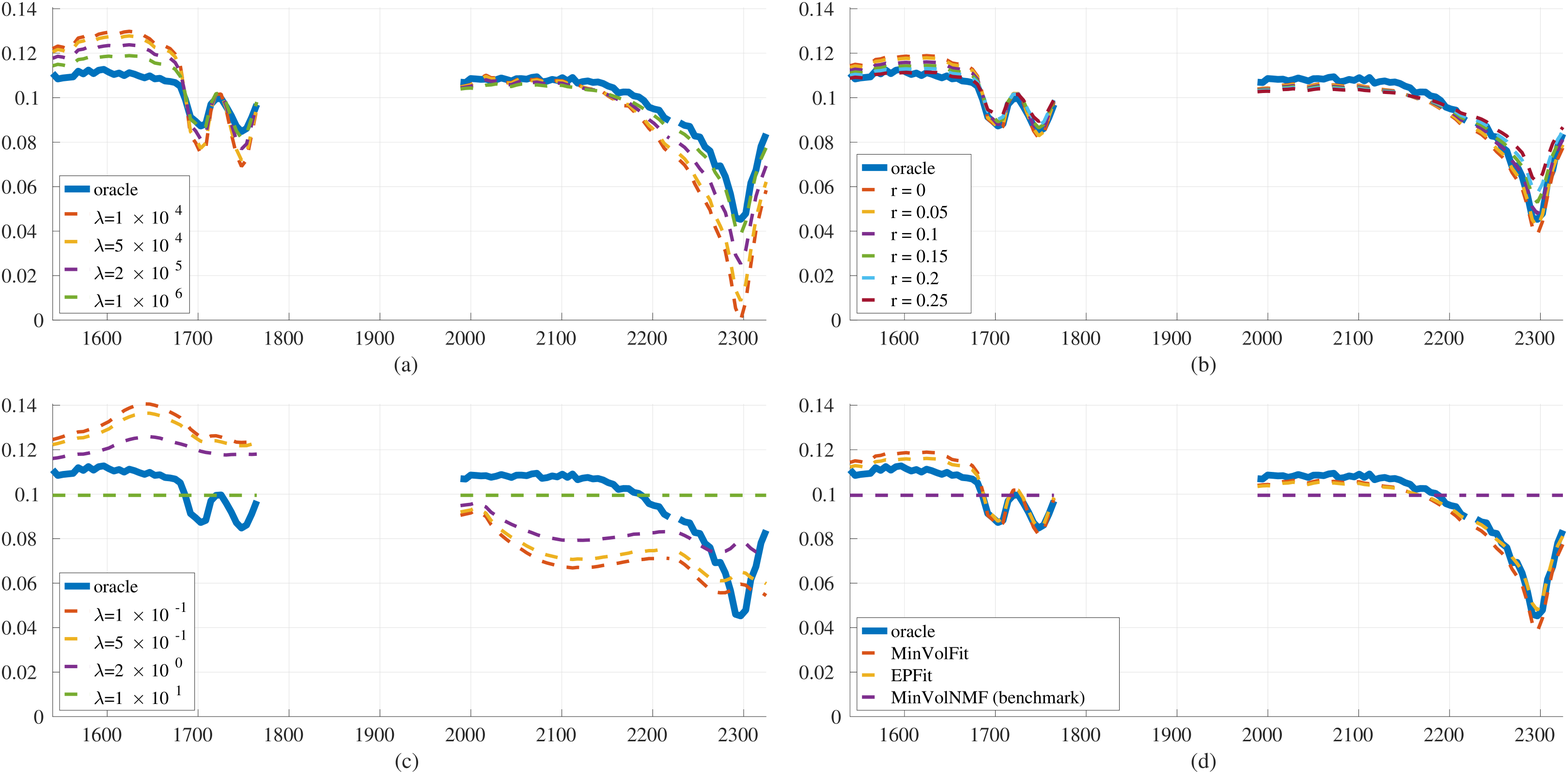}
    \caption{Recovered foreground material signatures of a silicone deposit on both ceramic tile and plywood backgrounds obtained using (a) the proposed $\mathtt{MinVolFit}$ algorithm with various choices of regularization weight $\lambda$; (b) the proposed $\mathtt{EPFit}$ algorithm with various choices of endpoint removal ratio $r$; and (c) the $\mathtt{MinVolNMF}$ benchmark algorithm with various choices of regularization weight $\lambda$. (d) Collection of most accurate estimated foreground material signatures (by smallest angular difference from oracle signature) obtained with each method.}
    \label{fig:silicone_on_ceramic_plywood}
\end{figure*}

\begin{table}[t]
\renewcommand{\arraystretch}{1.25}
\centering
\caption{Total runtime and average per-iteration runtime for various algorithms and benchmarks.}
\begin{tabular}{|c|c|c|c|}
\hline
Algorithm & $\lambda$ or $r$ & Total Time (s) & Avg. Per-Iter. Time (s) \\
\hline
\multirow{4}{*}{$\mathtt{MinVolFit}$} & $10^4$ & $2.217 \times 10^2$ & \multirow{4}{*}{$2.508 \times 10^{-1}$} \\
\cline{2-3}
& $5 \times 10^4$ & $2.479 \times 10^2$ & \\
\cline{2-3}
& $2 \times 10^5$ & $3.111 \times 10^2$ & \\
\cline{2-3}
& $10^6$ & $2.239 \times 10^2$ & \\
\hline
$\mathtt{EPFit}$ & all & $2.454 \times 10^2$ & $2.454 \times 10^{-1}$ \\
\hline
\multirow{4}{*}{$\mathtt{MinVolNMF}$} & 0.1 & $3.162 \times 10^1$ & \multirow{4}{*}{$3.172 \times 10^{-2}$} \\
\cline{2-3}
& 0.5 & $3.097 \times 10^1$ & \\
\cline{2-3}
& 2 & $3.264 \times 10^1$ & \\
\cline{2-3}
& 10 & $3.168 \times 10^1$ & \\
\hline
\end{tabular}
\end{table}

\noindent {\bf Results and Analysis:} 
Results for foreground signature extraction of silicone deposit on ceramic tile and plywood are shown in Figure~\ref{fig:silicone_on_ceramic_plywood}. We observe that for all three of our algorithms, there is a choice of hyperparameter that yields a foreground signature estimate that is close to the reference signature. In contrast, both benchmarks show some difference between the estimates and reference. In particular, benchmark 1 shows significant deviation from the reference signature. This is expected, as the approach of benchmark 1 is to concatenate all patches and solve for endpoint vectors as a rank-2 NMF problem, an assumption which does not hold when we consider a scenario with multiple backgrounds.

\section{Conclusion}
\label{sec:conclusion}

In this paper, we explored the problem of foreground material signature extraction under an intimate mixing bag-of-patches model. The problem of non-uniqueness of the solution for the foreground signature was identified and the space of all feasible solutions was derived. Conditions and criteria under which identifiable solutions are guaranteed were suggested and proven. Several algorithms with identifiability guarantees were proposed based on the previously suggested criteria. Experiments on synthetic and real data demonstrated the capability of the proposed algorithms to obtain identifiable solutions in settings where existing methods do not succeed.

\appendix

\subsection{Proof of Property~\ref{prop:space_of_solutions}}
\label{app:pf_space_of_solutions}

\begin{proof}
Let $\mathcal{Y} = \{ \mat{Y}^{(1)}, \ldots, \mat{Y}^{(K)} \}$ be a set of patches following the bag-of-patches model in \eqref{eq:bag_of_patches_model} with a true foreground material signature $\vec{f} \in \mathds{R}_{++}^M$, and assume $\vec{1}, \vec{f}, \vec{f} \odot \vec{f}$ are linearly independent. We may write the $k$th patch as
\begin{equation}
    \mat{Y}^{(k)} = \diag(\vec{v}^{(k)}) \begin{bmatrix} \vec{f} & \vec{1} \end{bmatrix} \mat{C}^{(k)},
\end{equation}
where each $\vec{v}^{(k)} \in \mathds{R}_{++}^M$ is strictly positive and each $\mat{C}^{(k)} \in \mathds{R}_{+}^{2 \times N_k}$ is non-negative for $k = 1, 2, \ldots, K$. By definition of the bag-of-patches model in \eqref{eq:bag_of_patches_model}, each patch $\mat{Y}^{(k)}$ must be rank 2. Suppose there exists an alternative solution $\vec{\tilde{f}}$ to the bag-of-patches model for the same set of patches $\mathcal{Y}$. That is, there exists $\vec{\tilde{v}}^{(k)} \in \mathds{R}_{++}^M$ and $\mat{\tilde{C}}^{(k)} \in \mathds{R}_{+}^{2 \times N_k}$ for $k = 1, 2, \ldots, K$ such that
\begin{equation}
    \mat{Y}^{(k)} = \diag(\vec{\tilde{v}}^{(k)}) \begin{bmatrix} \vec{\tilde{f}} & \vec{1} \end{bmatrix} \mat{\tilde{C}}^{(k)}.
\end{equation}
For each patch $k = 1, 2, \ldots, K$, it holds that
\begin{equation} \label{eq:proof1_eq1}
    \diag(\vec{v}^{(k)}) \begin{bmatrix} \vec{f} & \vec{1} \end{bmatrix} \mat{C}^{(k)} = \diag(\vec{\tilde{v}}^{(k)}) \begin{bmatrix} \vec{\tilde{f}} & \vec{1} \end{bmatrix} \mat{\tilde{C}}^{(k)}.
\end{equation}
Right-multiplying with the right-psuedoinverse $\mat{(\tilde{C}}^{(k)})^\dagger$ of $\mat{\tilde{C}}^{(k)}$ gives
\begin{equation} \label{eq:proof_prop_eq2}
    \diag(\vec{v}^{(k)}) \begin{bmatrix} \vec{f} & \vec{1} \end{bmatrix} \mat{C}^{(k)} \mat{(\tilde{C}}^{(k)})^\dagger = \diag(\vec{\tilde{v}}^{(k)}) \begin{bmatrix} \vec{\tilde{f}} & \vec{1} \end{bmatrix}.
\end{equation}
Note that if $\mat{Y}^{(k)}$ is rank 2, then each of $\diag(\vec{v}^{(k)})$, $\begin{bmatrix} \vec{f} & \vec{1} \end{bmatrix}$, and $\mat{C}^{(k)}$ are at least rank 2. Then $\mat{C}^{(k)}$ is full row rank. This holds similarly for $\mat{\tilde{C}}^{(k)}$. Then $\mat{C}^{(k)} \mat{(\tilde{C}}^{(k)})^\dagger = \begin{bmatrix} \alpha & \gamma \\ \beta & \delta \end{bmatrix}$ is rank 2. Thus, $\det(\mat{C}^{(k)} \mat{(\tilde{C}}^{(k)})^\dagger) = \alpha \delta - \beta \gamma \neq 0$.

The entry-wise ratio of the first and second columns on each side of \eqref{eq:proof_prop_eq2} yields
\begin{equation}
    \vec{\tilde{f}} = (\alpha \vec{f} + \beta \vec{1}) \oslash (\gamma \vec{f} + \delta \vec{1}).
\end{equation}
Substituting this definition of $\vec{\tilde{f}}$ in \eqref{eq:proof1_eq1} yields
\begin{equation}
\begin{aligned}
    & \diag(\vec{v}^{(k)}) \begin{bmatrix} \vec{f} & \vec{1} \end{bmatrix} \mat{C}^{(k)} \\
    &\qquad = \diag(\vec{\tilde{v}}^{(k)}) \begin{bmatrix} (\alpha \vec{f} + \beta \vec{1}) \oslash (\gamma \vec{f} + \delta \vec{1}) & \vec{1} \end{bmatrix} \mat{\tilde{C}}^{(k)}.
\end{aligned}
\end{equation}
Rearranging the terms on the RHS gives
\begin{equation} \label{eq:proof_prop_eq5}
\begin{aligned}
    & \diag(\vec{v}^{(k)}) \begin{bmatrix} \vec{f} & \vec{1} \end{bmatrix} \mat{C}^{(k)} \\
    &\qquad = \diag(\vec{\tilde{v}}^{(k)} \oslash (\gamma \vec{f} + \delta \vec{1})) \begin{bmatrix} \vec{f} & \vec{1} \end{bmatrix} \begin{bmatrix} \alpha & \gamma \\ \beta & \delta \end{bmatrix} \mat{\tilde{C}}^{(k)}.
\end{aligned}
\end{equation}
Using that each side is rank 2, it holds that their spans are equal. From this, we have the following system of equations:
\begin{equation} \label{eq:proof_prop_eq3}
\begin{aligned}
    (\vec{\tilde{v}}^{(k)} \odot \vec{f}) \oslash (\gamma \vec{f} + \delta \vec{1}) &= c_1 \vec{v}^{(k)} \odot \vec{f} + c_2 \vec{v}^{(k)} \\
    \vec{\tilde{v}}^{(k)} \oslash (\gamma \vec{f} + \delta \vec{1}) &= c_3 \vec{v}^{(k)} \odot \vec{f} + c_4 \vec{v}^{(k)}
\end{aligned}
\end{equation}
Substitution of $\vec{\tilde{v}}^{(k)} \oslash (\gamma \vec{f} + \delta \vec{1})$ gives
\begin{equation} \label{eq:proof_prop_eq4}
    c_3 \vec{v}^{(k)} \odot (\vec{f} \odot \vec{f}) + c_4 \vec{v}^{(k)} \odot \vec{f} = c_1 \vec{v}^{(k)} \odot \vec{f} + c_2 \vec{v}^{(k)}.
\end{equation}
Using $\vec{1}, \vec{f}, \vec{f} \odot \vec{f}$ independent, \eqref{eq:proof_prop_eq4} holds iff $c_2=c_3=0$ and $c_1=c_4$. Let $\epsilon_k = c_1=c_4$. Applying these conditions to \eqref{eq:proof_prop_eq3} and rearranging yields
\begin{equation}
    \vec{\tilde{v}}^{(k)} = \epsilon_k (\gamma \vec{f} + \delta \vec{1}) \odot \vec{v}^{(k)}.
\end{equation}
Substituting $\vec{\tilde{v}}^{(k)}$ in \eqref{eq:proof_prop_eq5} yields
\begin{equation}
    \diag(\vec{v}^{(k)}) \begin{bmatrix} \vec{f} & \vec{1} \end{bmatrix} \mat{C}^{(k)} = \diag(\vec{v}^{(k)}) \begin{bmatrix} \vec{f} & \vec{1} \end{bmatrix} \epsilon_k \begin{bmatrix} \alpha & \gamma \\ \beta & \delta \end{bmatrix} \mat{\tilde{C}}^{(k)}.
\end{equation}
From here, it is clear that $\mat{\tilde{C}}^{(k)} = \frac{1}{\epsilon_k} \begin{bmatrix} \alpha & \gamma \\ \beta & \delta \end{bmatrix}^{-1} \mat{C}^{(k)}$.
\end{proof}

\subsection{Proof of Theorem~\ref{thm:solutions_to_minvol}}
\label{app:pf_solutions_to_minvol}

\noindent We use the following result in the proof of Theorem~\ref{thm:solutions_to_minvol}:
\begin{lemma} \label{lem:support_vol}
Assume $\vec{s}=c (t \vec{v} + (1-t)\vec{1}) \oslash (u\vec{v} + (1-u)\vec{1})$ with $c > 0$ and $\vec{v}$ strictly positive such that $\vec{v} \odot \vec{v}, \vec{v}, \vec{1}$ are independent. Let $K_1(\vec{s}) = \|\vec{s}\|^2 \|\vec{1}\|^2 - (\vec{s}^\trans \vec{1})^2$ and $K_2(\vec{s}) = -(\|\vec{s}\|^4 - (\vec{s}^\trans (\vec{s} \odot \vec{s}))(\vec{s}^\trans \vec{1}))$.
If $t-u \neq 0$, then
\begin{enumerate}
    \item $\vec{1}$, $\vec{s}$, and $\vec{s} \odot \vec{s}$ are linearly independent and 
    \item $K_1(\vec{s})$ and $K_2(\vec{s})$ are strictly positive.
\end{enumerate}
\end{lemma}
\noindent The proof of Lemma~\ref{lem:support_vol} is given in Appendix~\ref{app:pf_support_vol}. We proceed with the proof of Theorem~\ref{thm:solutions_to_minvol}:

\begin{proof}
Let $\mathcal{S}$ be the set of all feasible solutions as given in Proposition~\ref{cor:patch_independent_solution}. For every feasible solution $\vec{s} \in \mathcal{S}$, there must exist parameters $\alpha, \beta, \gamma, \delta$ such that the following hold:
\begin{equation} \label{eq:S_in_abgd}
\begin{gathered}
    \vec{s} = (\alpha \vec{f} + \beta \vec{1}) \oslash (\gamma \vec{f} + \delta \vec{1}), \\
    \frac{\delta -\gamma r_a}{\alpha \delta - \beta\gamma} \geq 0, \enspace \frac{\delta r_b - \gamma}{\alpha \delta - \beta\gamma} \geq 0, \\
    \frac{\alpha - \beta r_b}{\alpha \delta - \beta\gamma} \geq 0, \enspace \frac{\alpha r_a - \beta}{\alpha \delta - \beta\gamma} \geq 0, \\
    \alpha \max f_i + \beta > 0, \enspace \alpha \min f_i + \beta > 0, \\
    \gamma \max f_i + \delta > 0, \enspace \gamma \min f_i + \delta > 0, \\
    \text{and} \enspace \alpha \delta - \beta \gamma \neq 0,
\end{gathered}
\end{equation}
where $r_a$ and $r_b$ are defined as in \eqref{eq:r}. The volume-minimization problem in Definition~\ref{def:minimum_volume_solution}, i.e., the volume minimization of \eqref{eq:volume_measure} may be replaced with the following maximization problem:
\begin{equation} \label{eq:min_vol_objective}
    \max_{\vec{s} \in \mathcal{S}} \quad g(\vec{s}) = \frac{(\vec{s}^\trans \vec{1})}{\|\vec{s}\|_2}.
\end{equation}

\noindent \textbf{Reparameterization:}  
Before we proceed, we would like to point out that w.l.o.g.~we make the assumption that ${\min_i f_i < 1 < \max_i f_i}$. \footnote{Note that since $\vec{f}$ is linearly independent of $\vec{1}$, it cannot be constant and hence $\min_i f_i < \max_i f_i$. To ensure that $\vec{f}$ used in the proof satisfies $\min_i f_i < 1< \max_i f_i$, a scaling can be applied to the original $\vec{f}$ so that $\vec{s}$ is defined in terms of the scaled version of $\vec{f}$ without loss of generality. \label{ft:f_straddle_one}}
Consider the following reparameterization of the problem using 
\begin{equation}
    \vec{s}=c (t \vec{f} + (1-t) \vec{1}) \oslash (u \vec{f} + (1-u) \vec{1}).
\end{equation}
Every element in $\mathcal{S}$ has a representation in $(c,t,u)$. To show this, note that every element $\vec{s} \in \mathcal{S}$ can be parameterized by $(\alpha, \beta, \gamma, \delta)$. Consider
\begin{equation}
\begin{aligned}
    \vec{s} &= (\alpha \vec{f} + \beta \vec{1}) \oslash (\gamma \vec{f} + \delta \vec{1}) \\
    &= \frac{\gamma+\delta}{\alpha+\beta} \left( \frac{\alpha}{\alpha+\beta} \vec{f} + \frac{\beta}{\alpha+\beta} \vec{1} \right) \oslash \left( \frac{\gamma}{\gamma+\delta} \vec{f} + \frac{\delta}{\gamma+\delta} \vec{1} \right).
\end{aligned}
\end{equation}
Note that dividing by $\alpha+\beta$ and $\gamma+\delta$ is always well-defined. Suppose $\alpha+\beta=0$: then $\alpha \min f_i + \beta = \alpha (\min f_i - 1) \leq 0$ (by assumption that $\min f_i < 1$). This violates the constraint in \eqref{eq:S_in_abgd}, so $\alpha+\beta$ must be non-zero. This follows similarly for $\gamma+\delta$. Finally, taking $c = \frac{\gamma+\delta}{\alpha+\beta}$, $t = \frac{\alpha}{\alpha+\beta}$, and $u = \frac{\beta}{\alpha+\beta}$, we have that $\vec{s}$ has a representation in $(c,t,u)$.

The mapping $(c,t,u) \mapsto \vec{s}$ is bijective for feasible solutions (see Appendix~\ref{app:pf_bijective_mapping}). Further, the manifold produced by this mapping is differentiable everywhere. To show this, consider the Jacobian of $\vec{s}$ with respect to the parameter vector $\theta=[c,t,u]^T$:
\begin{equation}
    \frac{d\vec{s}}{d\theta^\trans} = \frac{1}{c(t-u)} \begin{bmatrix}
        \vec{s} \odot \vec{s} & \vec{s} & \vec{1}
    \end{bmatrix} \begin{bmatrix}
        0 & 0 & -1 \\
        t-u & c & c \\
        0 & -c^2 & 0
    \end{bmatrix}.
\end{equation}
The Jacobian is well-defined within the set of feasible solutions. \footnote{Simple algebraic manipulation of the Jacobian shows that the denominator $c(t-u)$ cancels out of each term in the Jacobian. The Jacobian is undefined only if $u f_i + (1-u) = 0$ for some $i$, which violates the feasibility constraints.} The determinant of the matrix on the right is $\det = c^2(t-u)$, which is strictly non-zero for all feasible $\vec{s}$ as $c > 0$ and $t-u \neq 0$, and therefore this matrix has full rank. The columns $\vec{s} \odot \vec{s}$, $\vec{s}$, and $\vec{1}$ in the matrix on the left are linearly independent, so this matrix also has full rank. Then the Jacobian has full rank in the set of feasible solutions. Thus, the mapping $(c,t,u) \mapsto \vec{s}$ is differentiable everywhere within the feasibility set of the maximization problem. In summary, the mapping $(c,t,u) \mapsto \vec{s}$ produces a differentiable manifold, and therefore we may recast the problem of maximization over $\vec{s}$ into a problem of maximization over $(c,t,u)$.

Using the new parameterization, the volume minimization problem can be recast as follows
\begin{equation} \label{eq:min_vol_objective:2}
\begin{aligned}
    \max_{c,t,u} & ~~g(\vec{s}) = \frac{(\vec{s}^\trans \vec{1})}{\|\vec{s}\|_2} \\
    \text{s.t.} &\quad \frac{1-u - u r_a}{c(t-u)} \geq 0, \enspace \frac{(1-u) r_b - u}{c(t-u)} \geq 0, \\
    &\quad \frac{c(t - (1-t) r_b)}{c(t-u)} \geq 0, \enspace \frac{c(t r_a - (1-t))}{c(t-u)} \geq 0, \\
    &\quad c(t \max f_i + (1-t)) > 0, \enspace c(t \min f_i + (1-t)) > 0, \\
    &\quad u \max f_i + (1-u) > 0, \enspace u \min f_i + (1-u) > 0, \\
    &\quad \text{and} \enspace c(t-u) \neq 0.
\end{aligned}
\end{equation}

\noindent \textbf{Finding local maxima:} Let $\theta = [c,t,u]^\trans$. To identify local maxima, we can consider the maximization over $c(t-u)>0$ and over $c(t-u)<0$ as two separate maximization problems and identify local maxima in each.

\noindent \underline{Case $c(t-u) > 0$}:  
If $c(t-u)>0$, then we can rewrite the problem as 
\begin{equation} \label{eq:min_vol_objective:3}
\begin{aligned}
    \max_\theta & ~~g(\vec{s}) = \frac{(\vec{s}^\trans \vec{1})}{\|\vec{s}\|_2} \\
    \text{s.t.} &\quad 1-u - u r_a \geq 0, \enspace (1-u) r_b - u \geq 0, \\
    &\quad c(t - (1-t) r_b) \geq 0, \enspace c(t r_a - (1-t)) \geq 0, \\
    &\quad c(t \max f_i + (1-t)) > 0, \enspace c(t \min f_i + (1-t)) > 0, \\
    &\quad u \max f_i + (1-u) > 0, \enspace u \min f_i + (1-u) > 0, \\
    & \quad c(t-u)>0.
\end{aligned}
\end{equation}
The constraints simplify to
\begin{equation} \label{eq:simplified_constraints}
\begin{gathered}
    c > 0,\\
     -\frac{1}{\max f_i-1} < u \leq \frac{r_b}{1+r_b}, \\
      \frac{1}{1+r_a}  \leq t  < \frac{1}{1-\min f_i}.
\end{gathered}
\end{equation}
The derivation of this simplified set of constraints is given in Appendix~\ref{app:simplified_constraints}. Note that $r_b / (1+r_b) < 1/(1+r_a)$, so $t>u$ and hence $c(t-u) > 0$ holds implicitly. An illustration of the feasibility set on $(c,t,u)$ is given in Figure~\ref{fig:constraint_and_gradient}. Specifically, the feasibility set in \eqref{eq:simplified_constraints} is the rectangle encompassing the blue arrows.

\begin{figure}[htbp]
    \centering
    \includegraphics[width=0.5\textwidth]{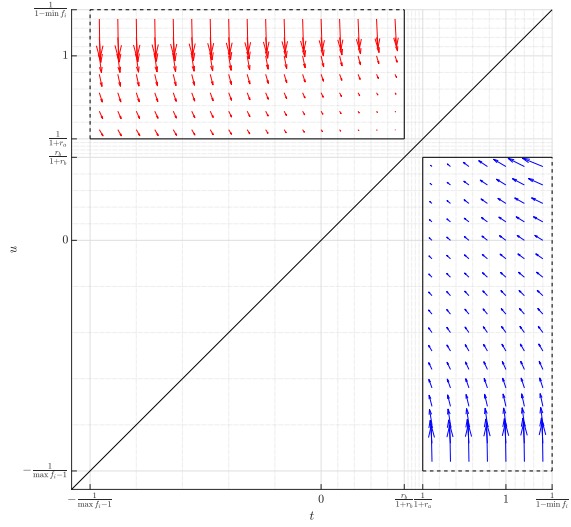}
    \caption{Visualization of feasibility sets and gradient for $\theta \in \Theta_1$ (blue arrows) and $\theta \in \Theta_2$ (red arrows) with a choice of fixed $c > 0$. the diagonal line indicates the non-feasible points with $t = u$. Note that the contiguous regions of feasible solutions lie strictly away from this line.}
    \label{fig:constraint_and_gradient}
\end{figure}

Denote the feasibility set containing all $\theta$ which satisfy the aforementioned constraints by $\Theta_1$. We can write the maximization as
\begin{equation}
\max_{\theta \in \Theta_1} J(c,t,u),
\end{equation}
where $J(c,t,u) = \frac{\vec{s}^\trans(c,t,u) \vec{1}}{\|\vec{s}(c,t,u)\|}$. \footnote{To find a local maximum, we need to find $c^*,t^*,u^*$ such that $J(c^*,t^*,u^*) \geq J(c,t,u)$ in the neighborhood $\| s(\theta^*)-s(\theta)\| \le \epsilon$ for some $\epsilon > 0$.} To show the existence of a local maximum, we consider the derivative of the objective with respect to $\theta$:
\begin{equation} \label{eq:gradient_minvol}
\begin{aligned}
    dJ/d\theta^T &= \frac{1}{(t-u)\|\vec{s}\|^3} [0,-cK_1(\vec{s}), \tfrac{1}{c} K_2(\vec{s})], \\
    K_1(\vec{s}) &= \|\vec{s}\|^2 \|\vec{1}\|^2 - (\vec{s}^\trans \vec{1})^2, \\
    K_2(\vec{s}) &= (\vec{s}^\trans (\vec{s} \odot \vec{s}))(\vec{s}^\trans \vec{1}) - \|\vec{s}\|^4.
\end{aligned}
\end{equation}
The computation of this derivative is given in Appendix~\ref{app:derivative}. Note that $c > 0$ and $c(t-u) > 0$ for any $\theta \in \Theta_1$, and therefore $t-u > 0$. Given that $t-u \neq 0$, it follows from Lemma~\ref{lem:support_vol} that $K_1(\vec{s})$ and $K_2(\vec{s})$ are both strictly positive. Based on the derivative, the objective is constant along $c$, monotonically decreasing along $t$, and monotonically increasing along $u$. Our feasibility set is defined by box constraints on $t$ and $u$, so a local maximum is obtained for any $c^*>0$ at the lower bound for $t$ and the upper bound for $u$. At any other point, it is possible to either move in decreasing $t$ or increasing $u$ and further increase the objective. Thus, a local maximum only exists at $(c,\frac{1}{1+r_a},\frac{r_b}{1+r_b})$ for any $c>0$. Substituting the parameters for a local maximum yields 
\begin{equation}
\vec{s}=\frac{c(1+r_b)}{1+r_a} ( \vec{f} + r_a \vec{1}) \oslash (r_b \vec{f} +  \vec{1}) = c' \vec{f}_0.
\end{equation}

\noindent \underline{Case $c(t-u) < 0$}: 
For the case of $c(t-u)<0$, we can rewrite (\ref{eq:min_vol_objective}) as
\begin{equation} \label{eq:min_vol_objective:4}
\begin{aligned}
    \max_\theta & ~~g(\vec{s}) = \frac{(\vec{s}^\trans \vec{1})}{\|\vec{s}\|_2} \\
    \text{s.t.} &\quad 1-u - u r_a \leq 0, \enspace (1-u) r_b - u \leq 0, \\
    &\quad c(t - (1-t) r_b) \leq 0, \enspace c(t r_a - (1-t)) \leq 0, \\
    &\quad c(t \max f_i + (1-t)) > 0, \enspace c(t \min f_i + (1-t)) > 0, \\
    &\quad u \max f_i + (1-u) > 0, \enspace u \min f_i + (1-u) > 0 \\
    & \quad c(t-u)<0
\end{aligned}
\end{equation}
Simplified constraints may be obtained similarly to the previous case, and are
\begin{equation} 
\begin{aligned}
    & c>0 ,\\
      \frac{1}{1+r_a}  \leq & u < \frac{1}{1-\min f_i}, \\
    -\frac{1}{\max f_i-1} < & t \leq \frac{r_b}{1+r_b} .
\end{aligned}
\end{equation}
An illustration of the feasibility set is given in Figure~\ref{fig:constraint_and_gradient} as the rectangle encompassing the red arrows. Note that the resulting set of constraints ensures $c(t-u)<0$. Denote the feasibility set containing all $\theta$ which satisfy the aforementioned constraints by $\Theta_2$. Note that $c > 0$ and $c(t-u) < 0$ for any $\theta \in \Theta_2$, and therefore $t-u < 0$. Following the approach used for $\Theta_1$, we have that the derivative is now (i.e., for $\Theta_2$) constant in $c$, increasing in $t$, and decreasing in $u$. Hence, a local maximum exists at every $c>0$, the right most point for $t$, and at the left most point for $u$, i.e., at $(c,\frac{r_b}{r_b+1},\frac{1}{r_a+1})$ for any $c>0$. Substituting the parameters for a local maximum yields\begin{equation}
\vec{s}=\frac{c(1+r_a)}{1+r_b} ( r_b \vec{f} +  \vec{1}) \oslash ( \vec{f} + r_a \vec{1}) = c' \vec{f}_0^{-1}.
\end{equation}
This completes the proof.
\end{proof}

\subsection{Proof of Theorem~\ref{thm:solutions_to_endfit}}
\label{app:pf_solutions_to_endfit}

\begin{proof}
Let $\mathcal{Y} = \{ \mat{Y}^{(1)}, \ldots, \mat{Y}^{(K)} \}$ be a set of patches following the bag-of-patches model in \eqref{eq:bag_of_patches_model} with a true foreground material signature $\vec{f} \in \mathds{R}_{++}^M$, and assume $\vec{1}, \vec{f}, \vec{f} \odot \vec{f}$ are linearly independent. We will begin by showing that any endpoint fit solution $\vec{f^*}$ satisfies $\vec{f^*} = c \vec{f_0}$ or $\vec{f^*} = c \vec{1} \oslash \vec{f_0}$. 

\vspace{\baselineskip}
\noindent \textbf{Forward direction:} Let $\vec{f^*}$ be an endpoint fit solution. According to Property~\ref{prop:space_of_solutions}, every coefficient matrix of a feasible solution has the form
\begin{equation}
    \mat{\tilde{C}}^{(k)} = \frac{1}{\epsilon_k} \begin{bmatrix}
        \alpha & \gamma \\
        \beta & \delta
    \end{bmatrix}^{-1} \mat{C}^{(k)}
\end{equation}
where $\epsilon_k > 0$ and $\alpha \delta - \beta \gamma \neq 0$. The $i$th column of the $k$th coefficient matrix is therefore
\begin{equation}
    \vec{\tilde{c}}_i^{(k)} = \frac{1}{\epsilon_k (\alpha \delta - \beta \gamma)} \begin{bmatrix}
            \delta c_{1i}^{(k)} - \gamma c_{2i}^{(k)} \\
            \alpha c_{2i}^{(k)} - \beta c_{1i}^{(k)}
    \end{bmatrix}.
\end{equation}
Note that the coefficient matrices $\mat{\tilde{C}}^{(k)}$ for $k = 1, \ldots, K$ must be non-negative (see Definition~\ref{def:solution}). Recall also that $\epsilon_k > 0$. Then for all $i$ and $k$, the following inequalities must hold:
\begin{equation} \label{eq:endpoint_inequalities}
\begin{aligned}
    \frac{\delta c_{1i}^{(k)} - \gamma c_{2i}^{(k)}}{\alpha \delta - \beta \gamma} &\geq 0, \\
    \frac{\alpha c_{2i}^{(k)} - \beta c_{1i}^{(k)}}{\alpha \delta - \beta \gamma} &\geq 0.
\end{aligned}
\end{equation}
Recall the definition of an endpoint fit solution: there exists indices $k_1$ and $k_2$ such that $\mat{\tilde{C}}^{(k_1)}$ contains a column $[x, 0]^\trans$ and $\mat{\tilde{C}}^{(k_2)}$ contains a column $[0, y]^\trans$ where $x,y > 0$. Thus, an endpoint fit solution must satisfy the inequalities in \eqref{eq:endpoint_inequalities} for every column of every coefficient matrix, and must satisfy the two equality conditions for some column(s). The set of endpoint fit solutions is therefore all choices of $\alpha, \beta, \gamma, \delta$ such that the following hold:
\begin{equation} \label{eq:set_of_epf_sol}
\begin{gathered}
    \frac{\delta c_{1i}^{(k)} - \gamma c_{2i}^{(k)}}{\alpha \delta - \beta \gamma} \geq 0 \quad \text{and} \quad \frac{\alpha c_{2i}^{(k)} - \beta c_{1i}^{(k)}}{\alpha \delta - \beta \gamma} \geq 0, \enspace \forall i,k, \\
    \exists i_1, k_1 \enspace \text{s.t.} \enspace \delta c_{1i_1}^{(k_1)} - \gamma c_{2i_1}^{(k_1)} = 0, \\
    \exists i_2, k_2 \enspace \text{s.t.} \enspace \alpha c_{2i_2}^{(k_2)} - \beta c_{1i_2}^{(k_2)} = 0.
\end{gathered}
\end{equation}

\noindent \underline{Case of $\alpha \delta - \beta \gamma > 0$}: Assume $\alpha \delta - \beta \gamma > 0$. The set defined by \eqref{eq:set_of_epf_sol} simplifies to
\begin{equation} \label{eq:set_of_epf_sol_pos}
\begin{gathered}
    \alpha c_{2i}^{(k)} - \beta c_{1i}^{(k)} \geq 0 \quad \text{and} \quad \delta c_{1i}^{(k)} - \gamma c_{2i}^{(k)} \geq 0, \enspace \forall i,k, \\
    \exists i_1, k_1 \enspace \text{s.t.} \enspace \alpha c_{2i_1}^{(k_1)} - \beta c_{1i_1}^{(k_1)} = 0, \\
    \exists i_2, k_2 \enspace \text{s.t.} \enspace \delta c_{1i_2}^{(k_2)} - \gamma c_{2i_2}^{(k_2)} = 0.
\end{gathered}
\end{equation}
Let $\mathcal{S}_1 = \{(i,k) \,|\, c_{1i}^{(k)} \neq 0\}$, and let $\mathcal{S}_1' = \{(i,k) \,|\, c_{1i}^{(k)} = 0\}$. The sets $\mathcal{S}_1$ and $\mathcal{S}_1'$ partition the set of all indices. The requirement that each $\mat{Y}^{(k)}$ is rank 2 implies that each $\mat{C}^{(k)}$ is rank 2, and therefore $\mathcal{S}_1$ is non-empty. \footnote{If $\mathcal{S}_1$ is empty, then $c_{1i}^{(k)} = 0$ and therefore every $\mat{C}^{(k)}$ is rank 1, which is a contradiction.} Further, the restriction that no column of $\mathcal{Y}^{(k)}$ for all $k$ is a zero-valued column implies that $c_{1i}^{(k)} \neq 0$ for all $(i,k) \in \mathcal{S}_1'$.  Consider the constraints on $\alpha$ and $\beta$, partitioned by $\mathcal{S}_1$ and $\mathcal{S}_1'$:
\begin{equation}
\begin{gathered}
    \beta \leq \alpha \tfrac{c_{2i}^{(k)}}{c_{1i}^{(k)}}, \enspace \forall (i,k) \in \mathcal{S}_1, \\
    0 \leq \alpha, \enspace \forall (i,k) \in \mathcal{S}_1', \\
    \exists (i_1, k_1) \in \mathcal{S}_1 \enspace \text{s.t.} \enspace \beta = \alpha \tfrac{c_{2i_2}^{(k_2)}}{c_{1i_2}^{(k_2)}} \enspace \text{or} \enspace \exists (i_1, k_1) \in \mathcal{S}_1' \enspace \text{s.t.} \enspace \alpha = 0.
\end{gathered}
\end{equation}
We may discard the case of $\exists (i_1, k_1) \in \mathcal{S}_1'$ s.t. $\alpha = 0$, as this implies $\beta \leq 0$ and therefore $\alpha \vec{f} + \beta \vec{1} \leq 0$, which contradicts \ref{enum:prop1_cond2} in Proposition~\ref{prop:feasible_space_solution}. Then we are left with the case of $\exists (i_1, k_1) \in \mathcal{S}_1$ s.t. $\beta = \alpha \tfrac{c_{2i_2}^{(k_2)}}{c_{1i_2}^{(k_2)}}$. For $\beta$ to be both a lower bound of the set $\{\alpha c_{2i}^{(k)} / c_{1i}^{(k)}\}_{\forall (i,k) \in \mathcal{S}_1}$ and equal to an element of the set, it must hold that
\begin{equation}
    \beta = \min_{i,k} \alpha \tfrac{c_{2i}^{(k)}}{c_{1i}^{(k)}},
\end{equation}
where $\tfrac{c_{2i}^{(k)}}{c_{1i}^{(k)}}$ is defined as $\infty$ if $c_{1i}^{(k)} = 0$. Note that if $\alpha = 0$, then $\beta = 0$ and therefore $\alpha \delta - \beta \gamma = 0$, which is a contradiction. Also, if $\alpha < 0$ then $\beta \leq 0$, and therefore $\alpha \vec{f} + \beta \vec{1} < 0$, which contradicts \ref{enum:prop1_cond2} in Proposition~\ref{prop:feasible_space_solution}. Then it must hold that $\alpha > 0$, and therefore an endpoint satisfies $\alpha = \beta \min_{i,k} \tfrac{c_{2i}^{(k)}}{c_{1i}^{(k)}} = \delta r_a$ (see \eqref{eq:r}).

Now, we will consider the constraints on $\gamma$ and $\delta$. Let $\mathcal{S}_2 = \{(i,k) \,|\, c_{2i}^{(k)} \neq 0\}$, and let $\mathcal{S}_2' = \{(i,k) \,|\, c_{2i}^{(k)} = 0\}$. The sets $\mathcal{S}_2$ and $\mathcal{S}_2'$ partition the set of all indices. Note that $\mathcal{S}_2$ is non-empty, and $c_{2i}^{(k)} \neq 0$ for all $(i,k) \in \mathcal{S}_2'$, by similar reasoning as above. The constraints on $\gamma$ and $\delta$, partitioned by $\mathcal{S}_1$ and $\mathcal{S}_1'$, may be expressed as
\begin{equation}
\begin{gathered}
    \gamma \leq \delta \tfrac{c_{1i}^{(k)}}{c_{2i}^{(k)}}, \enspace \forall (i,k) \in \mathcal{S}_2, \\
    0 \leq \delta, \enspace \forall (i,k) \in \mathcal{S}_2', \\
    \exists (i_2, k_2) \in \mathcal{S}_2 \enspace \text{s.t.} \enspace \gamma = \delta \tfrac{c_{1i_1}^{(k_1)}}{c_{2i_1}^{(k_1)}} \enspace \text{or} \enspace \exists (i_2, k_2) \in \mathcal{S}_2' \enspace \text{s.t.} \enspace \delta = 0.
\end{gathered}
\end{equation}
We may discard the case of $\exists (i_2, k_2) \in \mathcal{S}_2'$ s.t. $\delta = 0$, as this implies $\gamma \leq 0$ and therefore $\gamma \vec{f} + \delta \vec{1} \leq 0$, which contradicts \ref{enum:prop1_cond3} in Proposition~\ref{prop:feasible_space_solution}. Then we are left with the case of $\exists (i_2, k_2) \in \mathcal{S}_2$ s.t. $\gamma = \delta \tfrac{c_{1i_2}^{(k_2)}}{c_{2i_2}^{(k_2)}}$. For $\gamma$ to be both a lower bound of the set $\{\delta c_{1i}^{(k)} / c_{2i}^{(k)}\}_{\forall (i,k) \in \mathcal{S}_2}$ and equal to an element of the set, it must hold that
\begin{equation}
    \gamma = \min_{i,k} \delta \tfrac{c_{1i}^{(k)}}{c_{2i}^{(k)}},
\end{equation}
where $\tfrac{c_{1i}^{(k)}}{c_{2i}^{(k)}}$ is defined as $\infty$ if $c_{2i}^{(k)} = 0$. Note that if $\delta = 0$, then $\gamma = 0$ and therefore $\alpha \delta - \beta \gamma = 0$, which is a contradiction. Also, if $\delta < 0$ then $\gamma \leq 0$, and therefore $\gamma \vec{f} + \delta \vec{1} < 0$, which contradicts \ref{enum:prop1_cond3} in Proposition~\ref{prop:feasible_space_solution}. Then it must hold that $\delta > 0$. Therefore, an endpoint satisfies $\gamma = \delta \min_{i,k} \tfrac{c_{1i}^{(k)}}{c_{2i}^{(k)}} = \delta r_b$ (see \eqref{eq:r}).

Note that for $\beta = \alpha r_a$ and $\gamma = \delta r_b$, it holds that $\alpha \delta - \beta \gamma = \alpha \delta (1 - r_a r_b) > 0$. Finally, we have
\begin{equation}
    \vec{f^*} = (\alpha \vec{f} + \beta \vec{1}) \oslash (\beta \vec{f} + \gamma \vec{1}) = \frac{\alpha}{\delta} (\vec{f} + r_a \vec{1}) \oslash (r_b \vec{f} + \vec{1}) = c \vec{f_0}.
\end{equation}

\noindent \underline{Case of $\alpha \delta - \beta \gamma < 0$}: Assume $\alpha \delta - \beta \gamma < 0$. The set defined by \eqref{eq:set_of_epf_sol} simplifies to
\begin{equation} \label{eq:set_of_epf_sol_neg}
\begin{gathered}
     \beta c_{1i}^{(k)} - \alpha c_{2i}^{(k)} \geq 0 \quad \text{and} \quad \gamma c_{2i}^{(k)} - \delta c_{1i}^{(k)} \geq 0, \enspace \forall i,k, \\
    \exists i_1, k_1 \enspace \text{s.t.} \enspace \beta c_{1i_1}^{(k_1)} - \alpha c_{2i_1}^{(k_1)} = 0, \\
    \exists i_2, k_2 \enspace \text{s.t.} \enspace \gamma c_{2i_2}^{(k_2)} - \delta c_{1i_2}^{(k_2)} = 0.
\end{gathered}
\end{equation}
Arguing by symmetry, an endpoint satisfies $\alpha = \beta \min_{i,k} \frac{c_{1i}^{(k)}}{c_{2i}^{(k)}} = \beta r_b$ and $\delta = \gamma \min_{i,k} \frac{c_{2i}^{(k)}}{c_{1i}^{(k)}} = \gamma r_a$. Note that for $\alpha = \beta r_b$ and $\delta = \gamma r_a$, it holds that $\alpha \delta - \beta \gamma = \beta \gamma (r_a r_b - 1) < 0$. Finally, we have
\begin{equation}
    \vec{f^*} = (\alpha \vec{f} + \beta \vec{1}) \oslash (\beta \vec{f} + \gamma \vec{1}) = \frac{\beta}{\gamma} (r_b\vec{f} + \vec{1}) \oslash (\vec{f} + r_a\vec{1}) = c \vec{1} \oslash \vec{f_0}.
\end{equation}

\noindent \textbf{Backward direction:} Let $\vec{f^*} = c \vec{f_0} = (c \vec{f} + cr_a \vec{1}) \oslash (r_b \vec{f} + \vec{1})$. We will use Property~\ref{prop:space_of_solutions} to show that such a solution is an endpoint fit solution. A matching parameterization is $(\alpha, \beta, \gamma, \delta) = (c, cr_a, r_b, 1)$. From Property~\ref{prop:space_of_solutions}, the $i$th column of the $k$th coefficient matrix is
\begin{equation}
    \vec{\tilde{c}}_i^{(k)} = \frac{1}{\epsilon_k c(1-r_a r_b)} \begin{bmatrix}
            c_{1i}^{(k)} - r_b c_{2i}^{(k)} \\
            c c_{2i}^{(k)} - cr_a c_{1i}^{(k)}
    \end{bmatrix},
\end{equation}
where $\epsilon_k > 0$. Recall that $r_a := \min_{i,k} \frac{c_{2i}^{(k)}}{c_{1i}^{(k)}}$ and $r_b := \min_{i,k} \frac{c_{1i}^{(k)}}{c_{2i}^{(k)}}$. Let $(i_1, k_1) = \arg \min_{i,k} \frac{c_{2i}^{(k)}}{c_{1i}^{(k)}}$. Then
\begin{equation}
    c c_{2i_1}^{(k_1)} - cr_a c_{1i_1}^{(k_1)} = c (c_{2i_1}^{(k_1)} - \tfrac{c_{2i_1}^{(k_1)}}{c_{1i_1}^{(k_1)}} c_{1i_1}^{(k_1)}) = 0.
\end{equation}
Letting $(i_2, k_2) = \arg \min_{i,k} \frac{c_{1i}^{(k)}}{c_{2i}^{(k)}}$, it holds that
\begin{equation}
    c_{1i_2}^{(k_2)} - r_b c_{2i_2}^{(k_2)} = c_{1i_2}^{(k_2)} - \tfrac{c_{1i_2}^{(k_2)}}{c_{2i_2}^{(k_2)}} c_{2i_2}^{(k_2)} = 0.
\end{equation}
Thus, for solutions of the form $\vec{f^*} = c \vec{f_0}$, the conditions for an endpoint fit solution are satisfied. The argument for solutions of the form $\vec{f^*} = c \vec{1} \oslash \vec{f_0}$ follows similarly.
\end{proof}

\subsection{Proof of Lemma~\ref{lem:alternate_coefficient_constraint}}
\label{app:pf_alternate_coefficient_constraint}

\begin{proof}
Let $\mat{C}^{(k)} \in \mathds{R}_+^{2 \times N_k}$ for $k = 1, 2, \ldots, K$ be elementwise non-negative full row-rank matrices such that no column is equal to the zero vector. Let $\mat{C} = \begin{bmatrix} \mat{C}^{(1)} & \cdots & \mat{C}^{(K)} \end{bmatrix}$. For any invertible matrix $\mat{T} \in \mathds{R}^{2 \times 2}$, we have $\mat{T}^{-1} \mat{C}= \mat{T}^{-1} \begin{bmatrix} \mat{C}^{(1)} & \cdots & \mat{C}^{(K)} \end{bmatrix} =  \begin{bmatrix} \mat{T}^{-1} \mat{C}^{(1)} & \cdots & \mat{T}^{-1} \mat{C}^{(K)} \end{bmatrix}. $ Hence, if the matrix on the left $\mat{T}^{-1} \mat{C}$ is elementwise non-negative the matrix on the right is elementwise non-negative  or equivalently its submatrices $\mat{T}^{-1} \mat{C}^{(k)}$ for $k=1,2,\ldots,K$ are elementwise non-negative. Similarly, if we define the $i$th column of $\mat{C}$ as $\vec{c}_i$, then $\mat{C}=\begin{bmatrix} \vec{c}_1 & \cdots & \vec{c}_n \end{bmatrix} $ and consequently $\mat{T}^{-1} \mat{C}=  \begin{bmatrix} \mat{T}^{-1} \vec{c}_{1} & \cdots & \mat{T}^{-1} \vec{c}_{n} \end{bmatrix}.$ Elementwise non-negativity of the LHS implies the elementwise non-negativity of the RHS and vice versa, i.e., 
$\mat{T}^{-1} \mat{C} \geq 0$ if and only if  $\mat{T}^{-1} \vec{c}_i \geq 0, \forall i$. Thus,
\begin{equation}\label{eq:46}
    \mat{T}^{-1} \mat{C}^{(k)} \geq 0, \forall k \iff \mat{T}^{-1} \vec{c}_i \geq 0, \forall i.
\end{equation}
Consider the inequality on the RHS of (\ref{eq:46}). Since multiplication by a non-negative constant preserves the inequality, we have $\mat{T}^{-1} \vec{c}_i \ge 0 \iff \gamma_i \mat{T}^{-1} \vec{c}_i  \ge 0 \iff  \mat{T}^{-1}  \gamma_i \vec{c}_i \ge 0$ for all $i$, where $\gamma_i > 0$.
Using that no column is equal to the zero vector, we can define a scaling term $\gamma_i = 1/(c_{1i} + c_{2i})$ for the $i$th column for all $i$.  Using this choice of $\gamma_i$, we can introduce the scaled $i$ column $\vec{\bar{c}}_i=\gamma_i \vec{{c}}_i= [\tfrac{c_{1i}}{c_{1i}+c_{2i}},\tfrac{c_{2i}}{c_{1i}+c_{2i}} ]^T$. Therefore, we have
\begin{equation} \label{eq:lemma_1_ineq}
    \mat{T}^{-1} \vec{c}_i \geq 0, \forall i \iff \mat{T}^{-1} \vec{\bar{c}}_i 
    \geq 0, \forall i.
\end{equation}
Define $\alpha_i = c_{1i}/(c_{1i}+c_{2i})$; the scaled $i$th column may be expressed as $\vec{\bar{c}}_i = [\alpha_i, 1-\alpha_i]^\trans$. Let $a = \arg \max_i \alpha_i$ and $b = \arg \min_i \alpha_i$. Let $\bar{\alpha}_i = \tfrac{\alpha_i - \alpha_b}{\alpha_a - \alpha_b}$. Since every $\alpha_i \in [\alpha_a, \alpha_b]$, it holds that $\bar{\alpha}_i \in [0, 1]$. Using the definition of $\bar{\alpha}_i$, we may express each $\vec{\bar{c}}_i$ as
\begin{equation}
\begin{aligned}
    \vec{\bar{c}}_i = \begin{bmatrix}
        \alpha_i \\
        1 - \alpha_i
    \end{bmatrix} &= \bar{\alpha}_i \begin{bmatrix} \alpha_a \\ 1 - \alpha_a \end{bmatrix} + (1-\bar{\alpha}_i) \begin{bmatrix} \alpha_b \\ 1 - \alpha_b \end{bmatrix} \\
    &= \bar{\alpha}_i \vec{\bar{c}}_a + (1-\bar{\alpha}_i) \vec{\bar{c}}_b,
\end{aligned}
\end{equation}
i.e., as a convex combination of the vectors $\vec{\bar{c}}_a$ and $\vec{\bar{c}}_b$. Substituting $\bar{\vec{c}}_i$ into ... the RHS of \eqref{eq:lemma_1_ineq}, yields $\mat{T}^{-1} \vec{\bar{c}}_i = \bar{\alpha}_i \mat{T}^{-1} \vec{\bar{c}}_a + (1 - \bar{\alpha}_i) \mat{T}^{-1} \vec{\bar{c}}_b \geq 0, \forall i$. With $\bar{\alpha}_i$ and $1-\bar{\alpha}_i$ non-negative, the inequality holds if $\mat{T}^{-1} \vec{\bar{c}}_a \geq 0$ and $\mat{T}^{-1} \vec{\bar{c}}_b \geq 0$. Also, if $\mat{T}^{-1} \vec{\bar{c}}_i \geq 0$ for all $i$, then it holds for $i=a$ and $i=b$. Thus,
\begin{equation}
    \mat{T}^{-1} \vec{\bar{c}}_i \geq 0, \forall i
    \iff \mat{T}^{-1} \vec{\bar{c}}_a \geq 0 \enspace \text{and} \enspace \mat{T}^{-1} \vec{\bar{c}}_b \geq 0.
\end{equation}
Again, this holds for arbitrary positive scaling of each $\vec{\bar{c}}_i$. Note that $\alpha_a > 0$ from $\alpha_a \geq \alpha_i, \forall i$ and $\mat{C}$ having full row-rank. Similarly, $1-\alpha_b > 0$. Consider scaling $\vec{\bar{c}}_a$ by $\tfrac{1}{\alpha_a}$ and $\vec{\bar{c}}_b$ by $\tfrac{1}{1-\alpha_b}$. Then
\begin{equation}
\begin{aligned}
    \mat{T}^{-1} \vec{\bar{c}}_a \geq 0 &\iff \mat{T}^{-1} \begin{bmatrix}
        1 \\
        \tfrac{1-\alpha_a}{\alpha_a}
    \end{bmatrix} = \mat{T}^{-1} \begin{bmatrix}
        1 \\
        \tfrac{c_{2a}}{c_{1a}}
    \end{bmatrix} \geq 0, \\
    \mat{T}^{-1} \vec{\bar{c}}_b \geq 0 &\iff \mat{T}^{-1} \begin{bmatrix}
        \tfrac{\alpha_b}{1-\alpha_b} \\
        1
    \end{bmatrix} = \mat{T}^{-1} \begin{bmatrix}
        \tfrac{c_{1b}}{c_{2b}} \\
        1
    \end{bmatrix} \geq 0.
\end{aligned}
\end{equation}
Note that
\begin{equation} \label{eq:lem1_r_ordering}
\begin{aligned}
    \alpha_a \geq \alpha_i, \forall i &\iff \tfrac{c_{2a}}{c_{1a}} = \tfrac{1-\alpha_a}{\alpha_a} \leq \tfrac{1-\alpha_i}{\alpha_i} = \tfrac{c_{2i}}{c_{1i}}, \forall i, \\
    \alpha_b \leq \alpha_i, \forall i &\iff \tfrac{c_{1b}}{c_{2b}} = \tfrac{\alpha_b}{1-\alpha_b} \leq \tfrac{\alpha_i}{1-\alpha_i} = \tfrac{c_{1i}}{c_{2i}}, \forall i.
\end{aligned}
\end{equation}
Define $r' = \min_i \tfrac{c_{2i}}{c_{1i}}$ and $r'' = \min_i \tfrac{c_{1i}}{c_{2i}}$. From \eqref{eq:lem1_r_ordering} it holds that $r' = \tfrac{c_{2a}}{c_{1a}}$ and $r'' = \tfrac{c_{1b}}{c_{2b}}$. Thus,
\begin{equation}
    \mat{T}^{-1} \mat{C}^{(k)} \geq 0, \forall k \iff \mat{T}^{-1} \begin{bmatrix}
        1 \\ r'
    \end{bmatrix} \geq 0 \enspace \text{and} \enspace \mat{T}^{-1} \begin{bmatrix}
        r'' \\ 1
    \end{bmatrix} \geq 0.
\end{equation}
\end{proof}

\subsection{Proof of Lemma~\ref{lem:endpoint_fit}}
\label{app:pf_endpoint_fit}

\begin{proof}
For a matrix $\mat{A}$ and vectors $\vec{u}$, $\vec{v}$, define ${\mat{A} \subset \cone(\{\vec{u}, \vec{v}\})}$ to mean that the columns of $\mat{A}$ lie in the convex cone defined by the vectors $\vec{u}$ and $\vec{v}$.

Assume the conditions of the lemma hold. Let ${\mat{\tilde{Y}} = \begin{bmatrix} \diag(\vec{d}_1) \mat{Y}^{(1)} & \cdots & \diag(\vec{d}_K) \mat{Y}^{(K)} \end{bmatrix}}$. Note that the column space of $\mat{\tilde{Y}}$ is rank 2. Given that $\vec{\tilde{y}_i}$ and $\vec{\tilde{y}_j}$ are columns of $\mat{\tilde{Y}}$, and $\vec{\tilde{y}_i}$ and $\vec{\tilde{y}_j}$ have the minimum cosine similarity among columns of $\mat{\tilde{Y}}$, it follows that ${\mat{\tilde{Y}} \subset \cone(\{\vec{\tilde{y}_i}, \vec{\tilde{y}_j}\})}$ (see Appendix~\ref{app:min_cos_similarity}), i.e., columns of $\mat{\tilde{Y}}$ can be written as a nonnegative combination of $\vec{\tilde{y_i}}$ and $\vec{\tilde{y_j}}$. Then we may write
\begin{equation}
    \mat{\tilde{Y}} = \begin{bmatrix}
        \vec{\tilde{y}_i} & \vec{\tilde{y}_j}
    \end{bmatrix} \mat{\tilde{H}} = \diag(\vec{\tilde{y}_j}) \begin{bmatrix}
        \vec{\tilde{y}_i} \oslash \vec{\tilde{y}_j} & \vec{1}
    \end{bmatrix} \mat{\tilde{H}},
\end{equation}
where $\mat{\tilde{H}} \geq 0$. Given that $\vec{\tilde{y}_i}$ and $\vec{\tilde{y}_j}$ are columns of $\mat{\tilde{Y}}$, and must be linearly independent by $\mat{\tilde{Y}}$ rank 2, then the $i$th column of $\mat{\tilde{H}}$ is $[1, 0]^\trans$ and similarly the $j$th column is $[0, 1]^\trans$. Finally, using the partitioning ${\mat{\tilde{H}} = \begin{bmatrix} \mat{\tilde{C}}^{(1)} & \cdots & \mat{\tilde{C}}^{(K)} \end{bmatrix}}$, grouping related submatrices, and multiplying through by each $\diag(\vec{d_k})^{-1}$, we have
\begin{equation} \label{eq:lem2_bop}
    \mat{Y}^{(k)} = \diag(\vec{\tilde{y}_j} \oslash \vec{d_k}) \begin{bmatrix}
        \vec{\tilde{y}_i} \oslash \vec{\tilde{y}_j} & \vec{1}
    \end{bmatrix} \mat{\tilde{C}}^{(k)}, \enspace k = 1, \ldots, K.
\end{equation}
Since \eqref{eq:lem2_bop} satisfies the bag-of-patches model with ${\vec{f^*} = \vec{\tilde{y}_i} \oslash \vec{\tilde{y}_j}}$, and there exists columns $[x,0]^\trans$ and $[0,y]^\trans$ with $x,y > 0$ among the columns of $\mat{\tilde{C}}^{(k)}$ for $k = 1, \ldots, K$, then $\vec{f^*}$ is an endpoint fit solution.
\end{proof}

\subsection{Proof of Lemma~\ref{lem:support_vol}}
\label{app:pf_support_vol}

\begin{proof}
We begin by proving the linear independence of $\vec{1}$, $\vec{s}$, $\vec{s} \odot \vec{s}$ when $t-u \neq 0$. This is equivalent to $A\vec{1}+B\vec{s}+C\vec{s} \odot \vec{s} = 0$ having no non-trivial solutions. Substituting the expressions for $\vec{s}$ and its powers into the equation and gathering the equations by the powers of $\vec{v}$ yields
\begin{equation}
\begin{aligned}
    0 &= A \vec{1} + B c\frac{t\vec{v}+(1-t)\vec{1}}{u\vec{v}+(1-u)\vec{1}} + C c^2\left(\frac{t\vec{v}+(1-t)\vec{1}}{u\vec{v}+(1-u)\vec{1}}\right)^2 \\
    &= A (u^2\vec{v}^2 + 2u(1-u)\vec{v} + (1-u)^2\vec{1}) \\
    &\quad + Bc (tu\vec{v}^2 + (t(1-u)+(1-t)u)\vec{v} \\
    &\quad + (1-t)(1-u)\vec{1}) + Cc^2 (t^2\vec{v}^2 + 2t(1-t)\vec{v} + (1-t)^2\vec{1}) \\
    &= (u^2A + ctuB + c^2t^2C) \vec{v}^2 \\
    &\quad + (2u(1-u)A + c(t(1-u)+(1-t)u)B + 2c^2t(1-t)C) \vec{v} \\
    &\quad + ((1-u)^2A + c(1-t)(1-u)B + c^2(1-t)^2C) \vec{1}.
\end{aligned}
\end{equation}

\begin{equation}
    \begin{bmatrix}
        u^2 & ctu & c^2t^2 \\
        2u(1-u) & ct(1-u)+c(1-t)u & 2c^2t(1-t) \\
        (1-u)^2 & c(1-t)(1-u) & c^2(1-t)^2
    \end{bmatrix} \begin{bmatrix}
        A \\
        B \\
        C
    \end{bmatrix} = \begin{bmatrix}
        0 \\
        0 \\
        0
    \end{bmatrix}.
\end{equation}
If the matrix is non-singular, then the only solution is the trivial solution. To test whether the matrix is singular we examine its determinant, which simplifies to
\begin{equation}
    \text{det} = -c^3 (t-u)^3.
\end{equation}
The determinant is zero only if $t-u = 0$, which is a contradiction. Then $\vec{1}$, $\vec{s}$, and $\vec{s} \odot \vec{s}$ are linearly independent.

Next, we prove that  $K_1(\vec{s}) = \|\vec{s}\|^2 \|\vec{1}\|^2 - (\vec{s}^\trans \vec{1})^2$ is strictly positive. By Cauchy-Schwartz (CS) inequality, we have (i) $(\vec{s}^\trans \vec{1})^2 \leq \|\vec{s}\|^2 \|\vec{1}\|^2$ and (ii) $(\vec{s}^\trans \vec{1})^2 = \|\vec{s}\|^2 \|\vec{1}\|^2$ $\iff$ $\vec{s}$ and $\vec{1}$ are linearly dependent. From (i) we have $K_1(\vec{s} \geq 0$, and from (ii) and linear independence of $\vec{s}$ and $\vec{1}$ we have $K_1(\vec{s}) \neq 0$. Thus, $K_1(\vec{s} > 0$.

Similarly, we prove that $K_2(\vec{s}) =  -(\|\vec{s}\|^4 - (\vec{s}^\trans (\vec{s} \odot \vec{s}))(\vec{s}^\trans \vec{1}))$ is strictly positive. Setting $\vec{u}=\sqrt{\vec{s}}$ and $\vec{v}=\sqrt{\vec{s}}^3$ (defined element-wise) and using CS yields (i) $((\sqrt{\vec{s}})^\trans (\sqrt{\vec{s}}^{3}))^2 \leq \|\sqrt{\vec{s}}\|^2 \|\sqrt{\vec{s}}^3\|^2$ or alternatively $(\vec{s}^\trans \vec{s})^2  \leq (\vec{s}^\trans \vec{1})(\vec{s}^\trans (\vec{s} \odot \vec{s}))$, and (ii) $(\vec{s}^\trans \vec{s})^2  = (\vec{s}^T\vec{1})(\vec{s}^\trans (\vec{s} \odot \vec{s})) $ $\iff$ $\sqrt{\vec{s}} $ and $\sqrt{\vec{s}}^3 $ are linearly independent. From (i) we have $K_2(\vec{s}) \geq 0$, and from (ii) and linear independence of $\vec{1}$ and $\vec{s}$ we have $K_2(\vec{s}) \neq 0$. Thus, $K_2(\vec{s}) > 0$.
\end{proof}

\subsection{Projection onto intersection of non-negative orthant and unit sphere surface}
\label{app:projection}

The projection onto the intersection of the non-negative orthant and the surface of the unit sphere is given by the solution to the following problem:
\begin{equation}
    \vec{x}^* = \arg \min \norm{\vec{x} - \vec{y}}^2 \quad \text{s.t.} \quad \vec{x} \geq 0, \enspace \norm{\vec{x}}^2 = 1.
\end{equation}
Define $(\vec{y})_+$ to be the vector $\vec{y}$ with negative entries replaced with zero. Let $i^* = \arg \min_i |{y_i}|$ (selecting any entry if there are multiple minima). Let $\vec{e}_i$ be the $i$th canonical vector, i.e. a vector with zero-valued entries except for the one-valued $i$th entry. The solution for the projection is
\begin{equation}
    \vec{x}^* = \begin{cases}
        (\vec{y})_+ / \norm{(\vec{y})_+}, & \text{$\vec{y}$ has a positive entry}, \\
        \vec{e}_{i^*}, & \text{$\vec{y}$ is non-positive}.
    \end{cases}
\end{equation}

\subsection{Proof of bijective mapping}
\label{app:pf_bijective_mapping}

\begin{proof}
Let the set of feasible solutions $\mathcal{S}$ be defined as in \eqref{eq:S_in_abgd}. Also, w.l.o.g. assume $\min_i f_i < 1 < \max_i f_i$ (see footnote \ref{ft:f_straddle_one}).

\vspace{0.5\baselineskip}
\noindent \textbf{Injective mapping:} We will prove by contradiction. Suppose there exists $\vec{s_1} = c_1 (t_1 \vec{f} + (1-t_1) \vec{1}) \oslash (u_1 \vec{f} + (1-u_1) \vec{1})$ and $\vec{s_2} = c_2 (t_2 \vec{f} + (1-t_2) \vec{1}) \oslash (u_2 \vec{f} + (1-u_2) \vec{1})$ such that $\vec{s_1} = \vec{s_2}$ and $(c_1,t_1,u_1) \neq (c_2,t_2,u_2)$. Then
\begin{equation} \label{eq:s1=s2}
\begin{aligned}
    \vec{s_1} = \vec{s_2} \iff & (c_1 t_1 u_2 - c_2 t_2 u_1) \vec{f} \odot \vec{f} \\
    &+ (c_1 t_1(1-u_2) + c_1(1-t_1)u_2 \\
    &\quad - c_2 t_2(1-u_1) - c_2(1-t_2)u_1) \vec{f} \\
    &+ (c_1(1-t_1)(1-u_2) - c_2(1-t_2)(1-u_1)) \vec{1} \\
    &= \vec{0}.
\end{aligned}
\end{equation}
By linear independence of $\vec{1}$, $\vec{f}$, and $\vec{f} \odot \vec{f}$, it follows that \eqref{eq:s1=s2} holds if and only if the coefficients for each term $\vec{1}$, $\vec{f}$, and $\vec{f} \odot \vec{f}$ are equal to zero. This yields the system of equations
\begin{equation}
\begin{aligned}
    c_1 t_1 u_2 &= c_2 t_2 u_1, \\
    c_1 t_1(1-u_2) + c_1(1-t_1)u_2 &= c_2 t_2(1-u_1) + c_2(1-t_2)u_1, \\
    c_1(1-t_1)(1-u_2) &= c_2(1-t_2)(1-u_1).
\end{aligned}
\end{equation}
Expansion and substitution yields
\begin{equation}
\begin{aligned}
    c_1 t_1 u_2 &= c_2 t_2 u_1, \\
    c_1 t_1 + c_1 u_2 &= c_2 t_2 + c_2 u_1 \\
    c_1 &= c_2.
\end{aligned}
\end{equation}
Note that $c_1 = c_2 \neq 0$ in order for $\vec{s_1}$ and $\vec{s_2}$ to be feasible solutions. Then we may divide everywhere to remove $c_1$ and $c_2$, yielding
\begin{equation}
\begin{aligned}
    t_1 u_2 &= t_2 u_1, \\
    t_1 + u_2 &= t_2 + u_1 \\
    c_1 &= c_2.
\end{aligned}
\end{equation}
The only solutions for the above system of equations are
\begin{equation}
\begin{aligned}
    c_1 &= c_2, \\
    t_1 &= t_2, \\
    u_1 &= u_2,
\end{aligned} \quad \text{and} \quad \begin{aligned}
    c_1 &= c_2, \\
    t_1 &= u_1, \\
    t_2 &= u_2.
\end{aligned}
\end{equation}
Note that the latter solution implies $\vec{s_1} = \vec{s_2} \propto \vec{1}$, which is not in the feasible set $\mathcal{S}$. The remaining solution implies that $\vec{s_1}$ and $\vec{s_2}$ have the same parameterization, which is a contradiction. Thus, the mapping must be injective.

\vspace{0.5\baselineskip}
\noindent \textbf{Surjective mapping:} This property was shown in the proof of Theorem~\ref{thm:solutions_to_minvol}. It is repeated for here for reference.

Recall that every element $\vec{s} \in \mathcal{S}$ can be parameterized by $(\alpha, \beta, \gamma, \delta)$. Consider
\begin{equation}
\begin{aligned}
    \vec{s} &= (\alpha \vec{f} + \beta \vec{1}) \oslash (\gamma \vec{f} + \delta \vec{1}) \\
    &= \frac{\gamma+\delta}{\alpha+\beta} \left( \frac{\alpha}{\alpha+\beta} \vec{f} + \frac{\beta}{\alpha+\beta} \vec{1} \right) \oslash \left( \frac{\gamma}{\gamma+\delta} \vec{f} + \frac{\delta}{\gamma+\delta} \vec{1} \right).
\end{aligned}
\end{equation}
Note that dividing by $\alpha+\beta$ and $\gamma+\delta$ is always well-defined. Suppose $\alpha+\beta=0$: then $\alpha \min f_i + \beta = \alpha (\min f_i - 1) \leq 0$ (by assumption that $\min f_i < 1$). This violates the constraint in \eqref{eq:S_in_abgd}, so $\alpha+\beta$ must be non-zero. This follows similarly for $\gamma+\delta$. Finally, taking $c = \frac{\gamma+\delta}{\alpha+\beta}$, $t = \frac{\alpha}{\alpha+\beta}$, and $u = \frac{\beta}{\alpha+\beta}$, we have that $\vec{s}$ has a representation in $(c,t,u)$. This holds for all $\vec{s} \in \mathcal{S}$. Thus, the mapping is surjective.

\vspace{0.5\baselineskip}
\noindent \textbf{Bijective mapping:} We have shown that the mapping $(c,t,u) \mapsto \mathcal{S}$ is injective and surjective, and therefore the mapping is bijective.
\end{proof}

\subsection{Derivation of simplified constraints in \eqref{eq:simplified_constraints}}
\label{app:simplified_constraints}

\begin{proof}
Consider the constraints
\begin{equation}
\begin{gathered}
    1-u - u r_a \geq 0, \enspace (1-u) r_b - u \geq 0, \\
    c(t - (1-t) r_b) \geq 0, \enspace c(t r_a - (1-t)) \geq 0, \\
    c(t \max f_i + (1-t)) > 0, \enspace c(t \min f_i + (1-t)) > 0, \\
    u \max f_i + (1-u) > 0, \enspace u \min f_i + (1-u) > 0, \\
    c(t-u) > 0,
\end{gathered}
\end{equation}
where $r_a$ and $r_b$ are defined as in \eqref{eq:r}. First, we reorganize the constraints to obtain bounds on $c$, $t$, and $u$:
\begin{equation}
\begin{aligned}
    1 - u - ur_a \geq 0 &\iff \frac{1}{1+r_a} \geq u, \\
    (1-u)r_b - u \geq 0 &\iff \frac{r_b}{1+r_b} \geq u, \\
    c(t - (1-t) r_b) \geq 0 &\iff ct \geq c\frac{r_b}{1+r_b} \\
    &\iff \begin{cases}
        t \geq \frac{r_b}{1+r_b}, & c > 0, \\
        t \leq \frac{r_b}{1+r_b}, & c < 0, \\
        c = 0,
    \end{cases} \\
    c(t r_a - (1-t)) \geq 0 &\iff ct \geq c \frac{1}{1+r_a} \\
    &\iff \begin{cases}
        t \geq \frac{1}{1+r_a}, & c > 0, \\
        t \leq \frac{1}{1+r_a}, & c < 0, \\
        c = 0,
    \end{cases} \\
    c(t \max f_i + (1-t)) > 0 &\iff ct > -c \frac{1}{\max_i f_i - 1} \\
    &\iff \begin{cases}
        t > -\frac{1}{\max_i f_i - 1}, & c > 0, \\
        t < \frac{1}{\max_i f_i - 1}, & c < 0,
    \end{cases} \\
    c(t \min f_i + (1-t)) > 0 &\iff ct < c \frac{1}{1 - \min_i f_i} \\
    &\iff \begin{cases}
        t < \frac{1}{1 - \min_i f_i}, & c > 0, \\
        t > \frac{1}{1 - \min_i f_i}, & c < 0,
    \end{cases} \\
    u \max f_i + (1-u) > 0 &\iff u > -\frac{1}{\max_i f_i - 1}, \\
    u \min f_i + (1-u) > 0 &\iff u < \frac{1}{1 - \min_i f_i}.
\end{aligned}
\end{equation}
Consider the constraints on $u$:
\begin{equation}
\begin{gathered}
    u \leq \frac{1}{1+r_a}, \enspace u \leq \frac{r_b}{1+r_b}, \enspace u < \frac{1}{1 - \min_i f_i}, \\
    \text{and} \enspace -\frac{1}{\max_i f_i - 1} < u.
\end{gathered}
\end{equation}
Note that $1 < \frac{1}{1-\min_i f_i}$, while $\frac{1}{1+r_a} \leq 1$ and $\frac{r_b}{1+r_b} < 1$. Further, using that $r_a < r_b^{-1}$ it follows that
\begin{equation}
    \frac{\frac{r_b}{1+r_b}}{\frac{1}{1+r_a}} = (1+r_a) \frac{r_b}{1+r_b} = \frac{1 + r_a}{1 + r_b^{-1}} < 1.
\end{equation}
Thus, $\frac{r_b}{1+r_b} < \frac{1}{1+r_a}$. It is clear that $u \leq \frac{r_b}{1+r_b}$ is the strictest upper bound, so the constraints on $u$ simplify to
\begin{equation}
    -\frac{1}{\max_i f_i - 1} < u \leq \frac{r_b}{1+r_b}.
\end{equation}
Now, suppose $c > 0$. The constraints on $t$ are
\begin{equation}
\begin{gathered}
    \frac{r_b}{1+r_b} \leq t, \enspace \frac{1}{1+r_a} \leq t, \enspace -\frac{1}{\max_i f_i - 1} < t, \\ \text{and} \enspace t < \frac{1}{1 - \min_i f_i}.
\end{gathered}
\end{equation}
Note that $\frac{r_b}{1+r_b}$ and $\frac{1}{1+r_a}$ are non-negative, while $-\frac{1}{\max_i f_i - 1}$ is strictly negative. Using that $\frac{r_b}{1+r_b} < \frac{1}{1+r_a}$, the strictest lower bound is $\frac{1}{1+r_a} \leq t$. Thus, the constraints on $t$ when $c > 0$ become
\begin{equation}
    \frac{1}{1+r_a} \leq t < \frac{1}{1-\min_i f_i}.
\end{equation}
Suppose instead that $c < 0$. The constraints on $t$ are
\begin{equation}
\begin{gathered}
    t \leq \frac{r_b}{1+r_b}, \enspace t \leq \frac{1}{1+r_a}, \enspace t < -\frac{1}{\max_i f_i - 1}, \\
    \text{and} \enspace \frac{1}{1 - \min_i f_i} < t.
\end{gathered}
\end{equation}
Recall that $\frac{r_b}{1+r_b}$ and $\frac{1}{1+r_a}$ are non-negative, while $-\frac{1}{\max_i f_i - 1}$ is strictly negative; the strictest upper bound is $t < -\frac{1}{\max_i f_i - 1}$. However, the lower bound on $t$ is $\frac{1}{1-\min_i f_i} < t$, where $\frac{1}{1-\min_i f_i}$ is positive. It cannot hold that $t$ has a strict negative upper bound and a strict positive lower bound, so the case of $c < 0$ is non-feasible. In summary, the constraints are
\begin{equation}
\begin{gathered}
    c > 0, \\
    -\frac{1}{\max_i f_i - 1} < u \leq \frac{r_b}{1+r_b}, \\
    \frac{1}{1+r_a} \leq t < \frac{1}{1-\min_i f_i}.
\end{gathered}
\end{equation}
\end{proof}

\subsection{Computation of gradient in \eqref{eq:gradient_minvol}}
\label{app:derivative}

The objective is

The derivative of the objective with respect to $\vec{s}$ is
\begin{equation}
    \frac{dJ}{d\vec{s}} = \frac{\|\vec{s}\| \vec{1} - (\vec{s}^\trans \vec{1}) \|\vec{s}\|^{-1} \vec{s}}{\vec{s}^\trans \vec{s}} = \frac{(\vec{s}^\trans \vec{s}) \vec{1} - (\vec{s}^\trans \vec{1}) \vec{s}}{\|\vec{s}\|^3}.
\end{equation}
The derivative of $\vec{s}$ with respect to each parameter in $\theta = [c, t, u]^\trans$ is
\begin{equation}
\begin{aligned}
    \frac{d\vec{s}}{dc} &= \vec{s}/c, \\
    \frac{d\vec{s}}{dt} &= c (\vec{f} - \vec{1}) \oslash (u\vec{f} + (1-u)\vec{1}) \\
    &= \frac{1}{t-u} \frac{c (t-u) (\vec{f} - \vec{1})}{u\vec{f} + (1-u)\vec{1}} \\
    &= \frac{1}{t-u} \left( \frac{c (t\vec{f} - t\vec{1})}{u\vec{f} + (1-u)\vec{1}} - \frac{c (u\vec{f} - u\vec{1})}{u\vec{f} + (1-u)\vec{1}} \right) \\
    &= \frac{1}{t-u} \left( \frac{c (t\vec{f} + (1-t)\vec{1})}{u\vec{f} + (1-u)\vec{1}} - \frac{c (u\vec{f} + (1-u)\vec{1})}{u\vec{f} + (1-u)\vec{1}} \right) \\
    &= \frac{1}{t-u} \left( \vec{s} - c \vec{1} \right), \\
    \frac{d\vec{s}}{du} &= -\vec{s} \odot \frac{\vec{f} - \vec{1}}{u\vec{f} + (1-u)\vec{1}} \\
    &= -\vec{s} \odot \tfrac{1}{c} \frac{d\vec{s}}{dt} \\
    &= -\frac{1}{t-u} \left( \tfrac{1}{c} \vec{s} \odot \vec{s} - \vec{s} \right). \\
\end{aligned}
\end{equation}
Finally, the derivative of the objective with respect to the parameter vector $\theta$ is given by
\begin{equation}
\begin{aligned}
    \frac{dJ}{d\vec{s}^\trans} \cdot \frac{d\vec{s}}{dc} &= \frac{\|\vec{s}\| (\vec{1}^\trans \vec{s}) - (\vec{s}^\trans \vec{1}) \|\vec{s}\|}{\vec{s}^\trans \vec{s}} = 0, \\
    \frac{dJ}{d\vec{s}^\trans} \cdot \frac{d\vec{s}}{dt} &= \frac{1}{t-u} \frac{(\vec{s}^\trans \vec{s}) (\vec{1}^\trans \vec{s}) - (\vec{s}^\trans \vec{1}) (\vec{s}^\trans \vec{s})}{\|\vec{s}\|^3} \\
    &\quad - \frac{c}{t-u} \frac{(\vec{s}^\trans \vec{s}) (\vec{1}^\trans \vec{1}) - (\vec{s}^\trans \vec{1}) (\vec{s}^\trans \vec{1})}{\|\vec{s}\|^3} \\
    &= -\frac{c}{(t-u)\|\vec{s}\|^3} (\|\vec{s}\|^2 \|\vec{1}\|^2 - (\vec{s}^\trans \vec{1})^2) \\
    &= -\frac{c K_1(\vec{s})}{(t-u)\|\vec{s}\|^3}, \\
    \frac{dJ}{d\vec{s}^\trans} \cdot \frac{d\vec{s}}{du} &= -\frac{1/c}{t-u} \frac{(\vec{s}^\trans \vec{s}) (\vec{1}^\trans (\vec{s} \odot \vec{s})) - (\vec{s}^\trans \vec{1}) (\vec{s}^\trans (\vec{s} \odot \vec{s}))}{\|\vec{s}\|^3} \\
    &\quad + \frac{1}{t-u} \frac{(\vec{s}^\trans \vec{s}) (\vec{1}^\trans \vec{s}) - (\vec{s}^\trans \vec{1}) (\vec{s}^\trans \vec{s})}{\|\vec{s}\|^3} \\
    &= -\frac{1/c}{(t-u)\|\vec{s}\|^3} (\|\vec{s}\|^4 - (\vec{s}^\trans (\vec{s} \odot \vec{s}))(\vec{s}^\trans \vec{1})) \\
    &= \frac{K_2(\vec{s}) / c}{(t-u)\|\vec{s}\|^3}, \\
    \frac{dJ}{d\theta^\trans} &= \frac{1}{(t-u)\|\vec{s}\|^3} [0,-K_1(\vec{s}),+K_2(\vec{s})],
\end{aligned}
\end{equation}
where $K_1=\|\vec{s}\|^2 - (\vec{1}^\trans \vec{s})^2$ and $K_2= (\vec{s}^\trans (\vec{s} \odot \vec{s}))(\vec{s}^\trans \vec{1}) - \|\vec{s}\|^4$.

\subsection{Minimum cosine similarity implies basis of cone}
\label{app:min_cos_similarity}

\begin{lemma}
Let $\mat{A}$ be a rank-2 matrix with all positive entries, and suppose $\vec{a_i}$ and $\vec{a_j}$ are columns of $\mat{A}$ that have the smallest cosine similarity among all columns in $\mat{A}$. Then ${\mat{A} \subset \cone(\{\vec{a_i},\vec{a_j}\})}$.
\end{lemma}

\begin{proof}
Assume the conditions of the lemma hold. First, we show that the columns $\vec{a_i}$ and $\vec{a_j}$ must be linearly independent, by contradiction. Suppose $\vec{a_i}$ and $\vec{a_j}$ are linearly dependent: then the cosine similarity between $\vec{a_i}$ and $\vec{a_j}$ is 1. From the condition of the lemma, $\vec{a_i}$ and $\vec{a_j}$ have the smallest cosine similarity among all columns of $\mat{A}$. Then all pairs of columns of $\mat{A}$ have cosine similarity of 1, and therefore all columns of $\mat{A}$ are linearly dependent. This contradicts the condition that $\mat{A}$ is a rank 2 matrix. Thus, $\vec{a_i}$ and $\vec{a_j}$ are linearly independent.

Next, we show that every column of $\mat{A}$ can be expressed as a non-negative linear combination of the vectors $\vec{a_i}$ and $\vec{a_j}$. W.l.o.g., assume every column of $\mat{A}$ has a norm of one. \footnote{The cosine similarity measure of two vectors does not vary with the magnitude of the vectors. Additionally, the coefficients of the linear combination for an arbitrary column of $\mat{A}$ can be scaled in proportion to the magnitude of the column.} Let $\vec{a_k}$ be any column of $\mat{A}$. From the linear independence of $\vec{a_i}$ and $\vec{a_j}$, and using that $\mat{A}$ is rank 2, it follows that $\vec{a_k}$ can be expressed as a linear combination of $\vec{a_i}$ and $\vec{a_j}$. Denote this linear combination as $\vec{a_k} = \alpha \vec{a_i} + \beta \vec{a_j}$. The coefficients of this linear combination are given by the following formula:
\begin{equation}
\begin{aligned}
    \begin{bmatrix}
        \alpha \\
        \beta
    \end{bmatrix} &= (\mat{V}^\trans \mat{V})^{-1} \mat{V}^\trans \vec{a_k} = \begin{bmatrix}
        \vec{a_i}^\trans \vec{a_i} & \vec{a_i}^\trans \vec{a_j} \\
        \vec{a_j}^\trans \vec{a_i} & \vec{a_j}^\trans \vec{a_j}
    \end{bmatrix}^{-1} \begin{bmatrix}
        \vec{a_i}^\trans \vec{a_k} \\
        \vec{a_j}^\trans \vec{a_k}
    \end{bmatrix} \\
    &= \frac{1}{1-(\vec{a_i}^\trans \vec{a_j})^2} \begin{bmatrix}
        1 & -\vec{a_i}^\trans \vec{a_j} \\
        -\vec{a_i}^\trans \vec{a_j} & 1
    \end{bmatrix} \begin{bmatrix}
        \vec{a_i}^\trans \vec{a_k} \\
        \vec{a_j}^\trans \vec{a_k}
    \end{bmatrix} \\
    &= \frac{1}{1-(\vec{a_i}^\trans \vec{a_j})^2} \begin{bmatrix}
        \vec{a_i}^\trans \vec{a_k} - (\vec{a_i}^\trans \vec{a_j})(\vec{a_j}^\trans \vec{a_k}) \\
        \vec{a_j}^\trans \vec{a_k} - (\vec{a_i}^\trans \vec{a_j})(\vec{a_i}^\trans \vec{a_k})
    \end{bmatrix},
\end{aligned}
\end{equation}
where $\mat{V} = [\vec{a_i}, \vec{a_j}]$. Under the assumption that $\vec{a_i}$ and $\vec{a_j}$ have the smallest cosine similarity among all columns of $\mat{A}$, and all columns of $\mat{A}$ have norm 1, it holds that $\vec{a_i}^\trans \vec{a_k} \geq \vec{a_i}^\trans \vec{a_j}$ and $\vec{a_j}^\trans \vec{a_k} \geq \vec{a_i}^\trans \vec{a_j}$. Then the coefficients $\alpha$ and $\beta$ may be bounded below by
\begin{equation}
    \begin{bmatrix}
        \alpha \\
        \beta
    \end{bmatrix} \geq \frac{1}{1-(\vec{a_i}^\trans \vec{a_j})^2} \begin{bmatrix}
        (\vec{a_i}^\trans \vec{a_j})(1 - \vec{a_j}^\trans \vec{a_k}) \\
        (\vec{a_i}^\trans \vec{a_j})(1 - \vec{a_i}^\trans \vec{a_k})
    \end{bmatrix} \geq 0.
\end{equation}
This holds for any choice of $\vec{a_k}$. Thus, every column of $\vec{A}$ is a non-negative linear combination of the vectors $\vec{a_i}$ and $\vec{a_j}$, and therefore $\mat{A} \subset \cone(\{\vec{a_i}, \vec{a_j}\})$.
\end{proof}

\bibliographystyle{IEEEtran}
\bibliography{IEEEabrv,main}

\end{document}